\documentclass[aps,prx,onecolumn]{revtex4}
\usepackage{natbib}
\usepackage{color}
\usepackage{bm,amsmath,amssymb,amsfonts,epsfig,color,setspace,lscape%,cite
  ,mathptmx}
\usepackage{graphicx,float}
\usepackage{xcolor}
\usepackage[normalem]{ulem}
\usepackage[caption=false]{subfig}

\usepackage{mdframed}

\newcommand \bl{}

\newcommand \beq{\begin{eqnarray}}
\newcommand \eeq{\end{eqnarray}}
\newcommand \bit{\begin{itemize}}
\newcommand \eit{\end{itemize}}

\newcommand \nn{\nonumber}

\usepackage[normalem]{ulem}

\usepackage{hyperref}

\newcommand{\beginsupplement}{%
        \setcounter{table}{0}
        \renewcommand{\thetable}{S\arabic{table}}%
        \setcounter{figure}{0}
        \renewcommand{\thefigure}{S\arabic{figure}}%
        \setcounter{section}{0}
        \renewcommand{\thesection}{S\arabic{section}}%
        \setcounter{equation}{0}
        \renewcommand{\theequation}{S.\arabic{equation}}%
     }
\newtheorem{theorem}{Theorem}
\newtheorem{theo}{Instability Criterion}
\newmdtheoremenv{cor}{Corollary}

\newtheorem{definition}[theorem]{Definition}

\newtheorem{proposition}[theorem]{Proposition}
\newtheorem{rem}[theorem]{Remark}

\newenvironment{proof}[1][Proof]{\noindent\textbf{#1.} }{\ \rule{0.5em}{0.5em}}

\newcommand{\blue}{}
\DeclareMathAlphabet{\mathcal}{OMS}{cmsy}{m}{n}

\begin{document}

\title{From One Pattern into Another: Analysis of Turing Patterns in Heterogeneous Domains via WKBJ}
%Pattern Formation in Heterogeneous Media: WKBJ Meets Turing
%From One Pattern into Another: Pattern Formation in Heterogeneous Domains (VK: I would choose this one)
%Pattern Formation with Ambient Heterogeneity
\author{Andrew L. Krause$^{1,}$\footnote{krause@maths.ox.ac.uk}$^,$\footnote{These authors contributed equally to this work.}}
\author{V{\'a}clav Klika$^{2,}$\footnotemark[\value{footnote}]}
\author{Thomas E. Woolley$^{3}$}
\author{Eamonn A. Gaffney$^{1}$}

\affiliation{$^{1}$Wolfson Centre for Mathematical Biology, Mathematical Institute, University of Oxford, Andrew Wiles Building, Radcliffe Observatory Quarter, Woodstock Road, Oxford, OX2 6GG, United Kingdom}
\affiliation{$^{2}$Department of Mathematics, FNSPE, Czech Technical University in Prague, Trojanova 13, 120 00 Praha, Czech Republic}
\affiliation{$^{3}$Cardiff School of Mathematics, Cardiff University, Senghennydd Road, Cardiff, CF24 4AG, United Kingdom}
\begin{abstract}
Pattern formation from homogeneity is well-studied, but less is known concerning symmetry-breaking instabilities in heterogeneous media. It is nontrivial to separate observed spatial patterning due to inherent spatial heterogeneity from emergent patterning due to nonlinear instability. We employ WKBJ asymptotics to investigate Turing instabilities for a spatially heterogeneous reaction-diffusion system, and derive conditions for instability which are local versions of the classical Turing conditions We find that the structure of unstable modes differs substantially from the typical trigonometric functions seen in the spatially homogeneous setting. Modes of different growth rates are localized to different spatial regions. This localization helps explain common amplitude modulations observed in simulations of Turing systems in heterogeneous settings. We numerically demonstrate this theory, giving an illustrative example of the emergent instabilities and the striking complexity arising from spatially heterogeneous reaction-diffusion systems. Our results give insight both into systems driven by exogenous heterogeneity, as well as successive pattern forming processes, noting that most scenarios in biology do not involve symmetry breaking from homogeneity, but instead consist of sequential evolutions of heterogeneous states. The instability mechanism reported here precisely captures such evolution, and extends Turing's original thesis to a far wider and more realistic class of systems.\end{abstract}

\maketitle

\section{Introduction}
Since Alan Turing's celebrated work on morphogenesis \cite{turing1952chemical}, reaction-diffusion systems have been a paradigm of pattern formation throughout chemistry and biology \cite{de1991turing, cross1993pattern, kondo2010reaction, murray2004mathematical, green2015positional,woolley2014visions}. The most striking aspect of this theory is the emergence of heterogeneity from homogeneity. However, even Turing himself recognized this as an idealization when he wrote, ``Most of an organism, most of the time  is developing from one pattern into another, rather than from homogeneity into a pattern." Here, we concern ourselves with this heterogeneous setting, and determine the generalization of the Turing conditions to a reaction-diffusion system with explicit spatial dependence. We derive conditions for the instability of a heterogeneous steady state into a Turing-type pattern, with both the localization and structure of the pattern depending on the heterogeneity. Under a necessary hypothesis of a sufficiently slowly varying heterogeneous base state, our results clearly differentiate between spatial structure due to inherent spatial heterogeneity, and emergent patterns due to Turing-type instabilities. This then elucidates successive pattern formation in distinct stages.

This transition from one pattern into another has been noted as key in reconciling seemingly-divergent theories in morphogenesis \cite{green2015positional}. Turing's original theory was that his reaction-diffusion mechanism laid down a prepattern of heterogeneous morphogen concentration, which then drove cellular differentiation and morphogenesis directly (Fig.~\ref{RDvsPI}(a)-(c)). This is in contrast to theories of positional information (colloquially ``French-flag'' models) whereby cells \emph{a priori} are assigned locations relative to some developmental coordinate system, and perform different functions based on this positional information \cite{wolpert2016positional} (Fig.~\ref{RDvsPI}(d)). Spatial heterogeneity provides a way to reconcile these competing theories by allowing positional information to influence reaction-diffusion processes, leading to modulated patterns which are ubiquitous in nature (Fig.~\ref{RDvsPI}(d)-(f)). Additionally, heterogeneity permits successive reaction-diffusion patterning in stages, whereby patterning at different scales can arise (Fig.~\ref{RDvsPI}(e)-(g)). This is in line with work implicating chemical and cellular pre-patterns in developmental biology \cite{holloway1995reaction, warmflash2014method,Weber}, such as in the context of organising different regions along cell boundaries based on sharp variations in gene expression \cite{meinhardt1983cell, irvine2001boundaries}. %\textcolor{blue}{ALK: Further biological details could be added here; this is just a crude attempt at summarizing the two ways that heterogeneity solves problems in developmental biology..}

\begin{figure}
    \centering
    \includegraphics[width=0.9\textwidth]{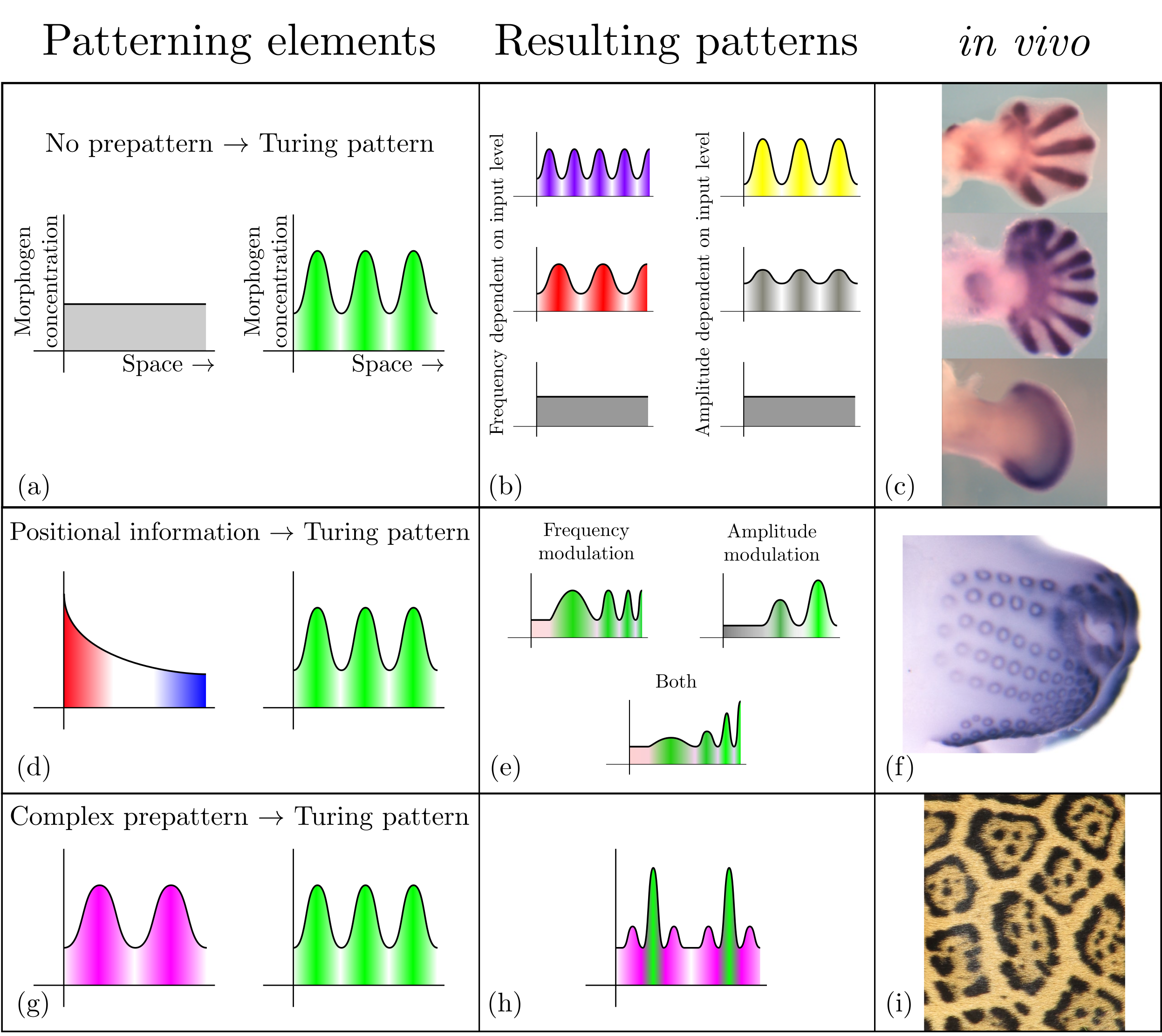}
    \caption{Different interactions of pattern formation mechanisms in development. (a) is a generic schematic of Turing pattern formation from homogeneity, with different pattern characteristics shown in (b), and, in (c), a biological example of a developing mouse paw in the presence of altered levels of Hox gene action. Positional information feeding into reaction-diffusion is shown in (d)-(e), consistent with observed structural characteristics of mouse whisker placodes in (f). Finally, successive reaction-diffusion patterning is shown in (g)-(h), with the example of Jaguar spots demonstrating large and small-scale pattern formation in (i). {\blue   In particular,  the schematic in (g) shows a sinusoidal prepattern (left peaks) feeding into a wave mode 3 Turing pattern (right peaks)   with, here for illustrative example,  the Turing pattern only able to   form  within the peaks of the prepattern. Thus, each peak forms a disjoint interval}. Mouse paw images from \citep{Sheth1476} R. Sheth, L. Marcon, M. F. Bastida, M. Junco, L. Quintana, R. Dahn, M. Kmita, J. Sharpe, M. A. Ros. Hox genes regulate digit patterning by controlling the wavelength of a Turing-type mechanism. Science, 338(6113):1476–1480, 2012. Reprinted with permission from AAAS. Mouse whisker placode image used with permission from Denis Headon. Jaguar picture by Jean Beaufort used under a CC0 Public Domain license from http://bit.ly/JaguarPicture.}
    \label{RDvsPI}
\end{figure}

Beyond theories of morphogenesis in developmental biology, models involving reaction-diffusion systems with spatial heterogeneity have been considered in many contexts. Examples include environmental heterogeneity in collective animal dispersal \cite{pickett1995landscape, clobert2009informed, cobbold2015diffusion, bassett2017continuous, kurowski2017two}, reaction-diffusion in domains with non-isotropic growth \cite{crampin2002pattern, krause2019influence}, as well as spatial invasion modelling \cite{sun2016pattern,belmonte2013modelling}, and models with differential diffusion leading to spatial inhomogeneity in plant root initiation \cite{brena2015stripe, avitabile2017spot}.  Spatial heterogeneity has been (numerically) observed to change local instability conditions for pattern formation \cite{benson1993diffusion, page2003pattern}, modulate size and wavelength of patterns \cite{page2005complex}, and localize (or pin) spike patterns in space \cite{iron2001spike, ward2002dynamics, wei2017stable}. We also note that the presence of even simple spatial heterogeneity can induce spatiotemporal behavior, such as changing the stability of patterned states and thereby inducing periodic movement of spike solutions \cite{krause2018heterogeneity, kolokolnikov2018pattern}. Bifurcation structures of reaction-diffusion equations with spatial heterogeneity have been considered for some time \cite{auchmuty1975bifurcation}. There is also a large literature on reaction-diffusion systems with strongly localized heterogeneities \cite{doelman2018pulse}, with \cite{stephetero} recently considering the case of a step-function heterogeneity in the reaction kinetics and deducing local Turing conditions on each side of the step. While we will also deduce local Turing conditions, we note that this limit is different from the case of smooth spatial heterogeneity we will consider here. 

Many experimental applications of reaction-diffusion systems have exploited an intuitive idea that a patterning instability is possible depending on the local environment, and, hence, one can think of \emph{local} pointwise Turing conditions in order to determine where patterning will occur \cite{lengyel1991modeling,epstein1996nonlinear,rudiger2003dynamics,miguez2005turing,rudiger2007theory}. This research has also given rise to various multiscale approaches for analysis of mode coupling between spatial forcing and emergent Turing patterns \cite{yang2002spatial,peter2005stripe,haim2015non}.  However, as far as we are aware, no justification for this localization, or the use of canonical (trigonometric) unstable eigenmodes, has been given in the literature. Several authors have attempted to deduce Turing conditions in spatially forced reaction-diffusion equations \cite{dewel1989effects, kuske1997pattern, otsuji2007mass}, but these results are limited to special cases regarding asymptotic assumptions and nonlinear kinetics, and even the case of varying diffusion coefficients is not perfectly understood \cite{mendez2010reaction}. We note that Dewel and Borckmans \cite{dewel1989effects} in particular analyse the case of slowly varying heterogeneity and employ a WBKJ-like ansatz, as we do below. However, their approach is substantially different from our own as they neglect the finite size of the domain, and do not recover the local Turing conditions that we seek, or the form of unstable modes.

Turing instabilities leading to pattern formation are typically considered to be induced due to the addition of diffusion (diffusion-driven instability) \cite{murray2004mathematical} and due to an increase in the domain-size \cite{klika2018domain}; below a certain critical domain size, patterns cannot be formed but above a minimal size, any small spatial perturbation of a reference homogeneous steady state will grow. The classical case focusing on spatially homogeneous systems is a textbook analysis and typically proceeds via a dispersion relation tying the Laplacian eigenmodes with the perturbation's growth rate \cite{murray2004mathematical,klika2017significance,maini2012turing,woolley2017turing}. However, as we shall show, justifying such a relationship between the growth rate and the operator's spectrum is much harder in the case of arbitrary spatial heterogeneity.

A major difficulty in analyzing instabilities in systems with spatial heterogeneity is that there is no simple generalization of Sturm--Liouville theory to multiple-component systems \cite{klika2018domain}. One can make use of the scalar theory when the heterogeneity appears in the same way in each component and is scaled such that the spatial operator, including diffusion, is identical in each equation. However, more generally, such a theory is difficult to use and, at best, one finds existence results, or must resort to numerical approaches \cite{dwyer1995eigenvalue, malamud2005uniqueness}. On the other hand, the WKBJ approximation has been employed in many optical and semi-classical quantum mechanical situations involving spatial heterogeneity \cite{bremmer1951wkb, alder1969quantal, griffiths2018introduction}, and, as we will demonstrate, has a straightforward generalization to coupled systems.

Here, we use WKBJ methods \cite{bender2013advanced} in order to compute instability criteria for a reaction-diffusion system with explicit spatial heterogeneity in the kinetics, under the assumption that the heterogeneity is sufficiently smooth and not rapidly varying compared with the diffusive length scales. Our analysis also shows several novel aspects of these instabilities in the presence of heterogeneity, such as modes supported in different regions of the domain depending on their growth rates. This phenomenon invalidates some heuristics commonly employed in homogeneous Turing pattern formation, such as restricting analysis to the mode with the fastest growth rate, which in the heterogeneous case varies across the domain. These structural results can help explain size and wavelength modulation in the presence of heterogeneity observed both in simulations and heterogeneous environments in experiment.

We begin by setting up the system and reviewing conditions for a Turing instability in the homogeneous case, and stating the corresponding conditions in the spatially heterogeneous setting in Section \ref{Overview}. This Section is a roadmap of our results and is intended to state the conditions without detailed derivation. Such a derivation is presented in Section \ref{derivation}, with the classical results in the homogeneous case in SI Section \ref{App.RederivationOfClassDDI}. We end this section with a discussion of properties of these solutions, and how their form implies the instability conditions, with some technical details in SI Section \ref{Sec.PropsOfWKB}.  In Section \ref{Illustration} we illustrate our results in the case of the Schnakenberg system, demonstrating both that our conditions for instability correspond to full numerical solutions, as well as showing various structural properties regarding the emergent unstable modes in line with our analysis. Finally, we discuss our results in Section \ref{Discussion}, highlighting both applications of our method and future directions for extensions. Someone interested primarily in our results, rather than the technicalities of the WKBJ calculations, can skip Section \ref{derivation}, and instead just read Sections \ref{Overview} and \ref{Illustration}-\ref{Discussion} to understand the implications of our results, as well as how to apply them to different systems.

\section{Homogeneous and Inhomogeneous Instability Conditions}\label{Overview}
Here, we state instability conditions for both homogeneous and heterogeneous two component reaction-diffusion systems which lead to emergent spatial patterning. In the heterogeneous setting we exploit asymptotically small diffusion coefficients, and so pose the general problem first. We consider a dimensional two component system in one spatial dimension,
$$ {\bf u}_t = {\bf D }_{dim} {\bf u}_{xx} + {\bf F}_{dim}({\bf u},x), ~~~ t>0, ~~~ x \in (0,L), ~~~ \nn {\bf D}_{dim}  = \left( \begin{array}{cc} D_1 & 0 \\ 0 & D_2 \end{array} \right), $$
where $D_1 > D_2 > 0$ are the diffusion coefficients and $L$ is the domain length. We prescribe Neumann boundary conditions ($ {\bf u }_x = 0 $ for $x\in\{0,L\}$) and the initial condition $ {\bf u }(x,0) = {\bf u }_0(x). $ We non-dimensionalise length scales with respect to $L$, time-scales with respect to a reaction time-scale $T$, and concentrations with respect to a typical concentration scale $U$ for both components and diffusion coefficients by $D_1$. Reusing ${\bf u}$, $t$ and $x$ to now represent non-dimensional quantities for brevity we have 
\beq \label{RDE}
{\bf u}_t = \epsilon^2{\bf D } {\bf u}_{xx} + {\bf F}({\bf u},x), ~~~ t>0, ~~~ x \in [0,1], ~~~~ \epsilon^2 = \frac{D_1T}{L^2},  ~~~~
{\bf D} = \left( \begin{array}{cc} 1 & 0 \\ 0 & d \end{array} \right) , 
~~~d = D_2/D_1 \leq  1,
\eeq
where ${\bf F}({\bf u},x) = (f({\bf u},x),g({\bf u},x)) $ is now a non-dimensional vector of kinetic functions. Below we assume $0< \epsilon \ll 1. $ This asymptotic assumption is not physiologically unreasonable in developmental settings. {\bl Consider kinetic timescales of $T\sim 10$ minutes, the shortest that would allow for gene expression \cite{tennyson1995human, singh2009rates}, a domain length of $L\sim 1$mm and a diffusion coefficient of $D_1 \sim 9.4\times 10^{-9}$cm$^2$s$^{-1}$, as measured in a synthetic biology experiment for the protein  Lefty, 
the most mobile of the prospective 
developmental Turing morphogen pair Nodal and Lefty, as part of an   investigation analysing Turing's mechanism in depth \cite{Sekine}. 
One then has $\epsilon^2 \sim 5.6 \times 10^{-4}$ and hence $\epsilon \sim 0.024$. }

Let ${\bf u}^*(x)$ denote a steady 
state, so that 
$$
 {\bf 0} = \epsilon^2 {\bf D } {\bf u}^*_{xx} + {\bf F}({\bf u}^*,x), ~~~ x \in [0,1], $$
with boundary conditions $ {\bf u }^*_x = 0$. To generalize the notion of a homogeneous steady state, we only consider the possibility that ${\bf u}^*$ oscillates with spatial derivatives of scale $O(1)$, or smaller, specifically excluding spatial oscillations on the scale of $O(1/\epsilon)$, or larger. Hence, 
${\bf u}^*$ is independent of $\epsilon$ and we have 
\begin{equation} \label{hetss} {\bf 0} =  {\bf F}({\bf u}^*,x) + O(\epsilon^2),\end{equation}
as long as the spatial heterogeneity in ${\bf F}$ permits $ {\bf u }^*_x = 0$ at $x=0,1$. If instead $ {\bf u }^*_x \neq 0$ at either boundary then a boundary layer with concomitant large derivatives will form, a possibility which we neglect in the subsequent analysis. {\bl However,  note that when the underlying heterogeneity is the result of prior patterning, as is the motivating background here, the steady   state ${\bf u}^*$  will automatically will be subject to the same zero flux  boundary conditions unless the boundary fluxes change as one patterning mechanism progresses into the next. However, zero flux boundary conditions are ubiquitous in models of biological pattern formation \cite{murray2004mathematical} and hence requiring $ {\bf u }^*_x = 0$ at $x=0,1$ does not constitute a particularly demanding constraint, at least in the context of hierarchical self-organisation in developmental biology.}

 %(numerically it is observed that such boundary effects have a width of $O(\epsilon)$ but that they can exhibit a Runge-phenomenon like jump in height).

%In developmental settings, the heterogeneous steady state given by Equation \eqref{hetss} can come about from one of at least two distinct physical mechanisms. Firstly, exogenous sources of morphogen or variations in developing tissue can lead to gradients of background production, influencing the local reaction kinetics across a domain. Secondly, a previous pattern-forming mechanism could have led to a broad wavelength pattern, and successive patterning on this heterogeneous domain could then occur via a Turing-type mechanism with a background heterogeneity (see Figure \ref{RDvsPI}). In the latter case, either the kinetic timescales between the successive reactions are separated, or the subsequent patterning mechanism involves cells (or other structures) with different diffusion coefficients from the chemical prepattern.

{\bl We proceed by linearising about this steady state via ${\bf w}={\bf u}-{\bf u}^*$, which is assumed small component-wise even relative to the scale of $\epsilon$, to yield}
\beq \label{weqn1}
{\bf w}_t =  \epsilon^2 {\bf D } {\bf w}_{xx} + {\bf J}(x){\bf w},
\eeq
where ${\bf J}(x)$ is the Jacobian matrix of ${\bf F}$ evaluated at ${\bf u}^*(x)$. System \eqref{weqn1} inherits homogeneous Neumann boundary conditions and the initial condition ${\bf w}(x,0)= {\bf u}(x,0)-{\bf u}^*(x)$. The fundamental impact of spatial heterogeneity in the  kinetics is that the Jacobian ${\bf J}$ possesses an explicit spatial dependence (and formally an additional $O(\epsilon^2)$ dependence, though we can neglect this via the asymptotic analysis going forward). The standard derivation in the homogeneous setting proceeds by assuming the ansatz ${\bf w} \propto e^{\lambda t}{\bf q}(x)$, justified by linearity. One then uses eigenvalues of the Laplacian to find $\lambda(n)$, where $n$ is a spectral parameter, resulting in conditions which imply $\Re(\lambda) > 0$ and, hence, instability. 

This approach does not generalize to the heterogeneous setting due to the explicit spatial dependence of ${\bf J}$, and so instead we think of varying $\lambda$ as a parameter and searching for eigenvalues consistent with the form of the solution when $\Re(\lambda)>0$. We first state a reformulation of the classical homogeneous conditions before generalizing to the heterogeneous case. We give a detailed rederivation in the homogeneous case in SI Section \ref{App.RederivationOfClassDDI}, arriving at the following formulation of the Turing conditions:

\begin{theo}[Homogeneous]\label{homogturing} Let $0<\epsilon \ll 1$, and ${\bf J}$ a constant matrix for all $x \in [0,1]$. If we assume stability to homogeneous perturbations, i.e.,
\beq \label{stab0} {\rm tr}( {\bf J}) < 0, ~~~~ {\rm det}({\bf J}) > 0,
\eeq
then there exists a non-homogeneous perturbation ${\bf w}$ satisfying \eqref{weqn1} which grows exponentially in time in the interval $x \in [0,1]$  if
\beq \label{instabk}  {\rm tr}( {\bf D}^{-1}{\bf J}) > 0, ~~~~ [{\rm tr}( {\bf D}^{-1}{\bf J})]^2-4 {\rm det}({\bf D}^{-1}{\bf J}) > 0.
\eeq 
\end{theo}

We find an analogous result in the spatially heterogeneous setting involving a much more complicated form of unstable modes explicitly depending on the growth rate $\lambda$, so that, to leading order, we have unstable solutions of the form  ${\bf w} \propto e^{\lambda t}{\bf q}(x,\lambda)$. Additionally, for different growth rates $\lambda$, the instability may be restricted to different subsets of the spatial domain (asymptotically at leading order). We will denote the largest of these regions, within which we anticipate patterns to be confined, as $\mathcal{T}_0$, which can consist of multiple disjoint intervals (as in Fig.~\ref{RDvsPI}(h)). We denote the interior of this region as $\mathcal{T}_0^i$. Conditions for instability in the heterogeneous case then follow from Criterion \ref{heteroturing},  and Proposition \ref{c1}, which are stated and derived in the next section. These conditions can be stated as:

\begin{theo}[Heterogeneous]\label{hettheo} Let $ 0<\epsilon \ll 1$, and assume that the quantity $[{\rm tr}( {\bf D}^{-1}{\bf J}(x))]^2-4 {\rm det}({\bf D}^{-1}{\bf J}(x))$ has only simple zeros for $x \in [0,1]$. If we assume stability to perturbations in the absence of diffusion, i.e.,
\beq {\rm tr}( {\bf J}(x)) < 0, ~~ {\rm det}({\bf J}(x)) > 0, ~~~ \text{for all }x \in [0,1],
\eeq
then there exists a non-homogeneous perturbation ${\bf w}$ satisfying \eqref{weqn1} {\bl (for sufficiently small $\epsilon$ and} to leading order in $\epsilon$) which grows exponentially in time for all $x \in \mathcal{T}_0^i$  if
\beq \label{tchet0}  {\rm tr}( {\bf D}^{-1}{\bf J}(x)) > 0, ~~ [{\rm tr}( {\bf D}^{-1}{\bf J}(x))]^2-4 {\rm det}({\bf D}^{-1}{\bf J}(x)) > 0, ~~~ \text{for all }x \in \mathcal{T}_0^i, \eeq
where $\mathcal{T}_0^i$ is the largest set for which conditions \eqref{tchet0} hold.
\end{theo}
 More generally, the conditions of Criterion \ref{hettheo} are
 exactly a local version of the homogeneous results, so that the same conditions satisfied on a subset of the full spatial domain imply a pattern forming instability on that subset. Both homogeneous and heterogeneous conditions hold for sufficiently small $\epsilon$, which can be thought of as a sufficiently large spatial domain. In this case, one can neglect the discrete wave mode selection, though we do give discrete dispersion relations in SI Section \ref{App.RederivationOfClassDDI} for the homogeneous case, and Criterion \ref{heteroturing} for the heterogeneous case. These discrete conditions give concrete ways to determine precisely which modes become unstable, and their associated growth rates, for a fixed value of $\epsilon$. In the next section we will describe how to derive Criterion \ref{hettheo} and these results mentioned above, and also further structural details about such instabilities which emerge from the form of unstable modes. One can skip these details on a first read and see an illustration of the results in Section \ref{Illustration}.

\section{Deriving the Spatially Inhomogeneous Conditions}\label{derivation}
We start our analysis of the heterogeneous Turing instability by analysing the stability of the steady state in line with the usual Turing instability analysis.  %The steady state is given by the root ${\bf u}^{*}(x)$ of ${\bf F}({\bf u}^{*},x)={\bf 0}$, accurate up to $O(\epsilon^2)$ corrections, and
%assumed to be twice-differentiable in $x$. A homogeneous perturbation satisfies 
%\beq \label{weqn11}
%{\bf w}_t = {\bf J}(x){\bf w}  + \epsilon^2 {\bf D } {\bf w}_{xx}  = {\bf J}(x){\bf w}  + O( \epsilon^2).
%\eeq
We seek solutions to Equation \eqref{weqn1} in the form ${\bf w}(x,t) = e^{\lambda t} {\bf q}(x)$ (as linearity permits separability in $t$ and $x$) to find
\beq \label{initialproblem}
{\bf 0}  = \epsilon^2 {\bf D }  {\bf q}_{xx} + {\bf J}_\lambda(x){\bf q}, 
\eeq
with ${\bf J}_\lambda = {\bf J} - \lambda {\bf I}$. We then proceed in direct analogy to the scalar  WKBJ expansion \cite{griffiths2018introduction, bender2013advanced}, with  
$$ {\bf q} = \exp\left[\frac{i\varphi(x)}\epsilon\right]{\bf p}(x), ~~~~~~~~ {\bf p}(x) = {\bf p}_0(x) + \epsilon {\bf p}_1(x) +\epsilon^2 {\bf p}_1(x)  + \ldots~~~. $$ 
Thus, with $'$ denoting the ordinary derivative with respect to $x$,
\beq {\bf q}_{xx} &=& \nonumber 
 \exp\left[\frac{i\varphi}\epsilon\right]\left[ 
\left( \frac{i\varphi'}\epsilon {\bf p} + {\bf p}'\right)' +  
\frac{i\varphi'}\epsilon\left( \frac{i\varphi'}\epsilon {\bf p} + {\bf p}'\right)  
\right] 
=
 \exp\left[\frac{i\varphi}\epsilon\right]\left[ 
-\frac{\varphi'^2}{\epsilon^2} {\bf p} + 
\frac 1 \epsilon\left(2i \varphi'{\bf p}'+ i \varphi'' {\bf p}\right) 
+ {\bf p}''
\right],
\eeq
and, hence,
$$ {\bf 0} = [-\varphi'^{2} {\bf D} +{\bf J}_\lambda] {\bf p} +
\epsilon \left[ 2i\varphi' {\bf D}{\bf p}' + i \varphi''{\bf D}{\bf p}\right] 
+O(\epsilon^2),%\epsilon^2 {\bf D}{\bf p}'',
$$
where the $O(\epsilon^2)$ terms from ${\bf J}$ were neglected, as we will not need to consider the second order problem below.
At leading order in $\epsilon$ we have
$$ {\bf 0} = [-\varphi'^{2} {\bf D} +{\bf J}_\lambda] {\bf p}_0 = 
{\bf D} \left\{[-\varphi'^{2} {\bf I} +{{\bf B}_\lambda}] {\bf p}_0 
\right\},
$$ 
where we define the matrix ${\bf B}_\lambda={\bf D}^{-1}{\bf J}_\lambda = {\bf D}^{-1}({\bf J}-\lambda{\bf I})$ and at next to leading order 
$$ {\bf 0} = [-\varphi'^{2} {\bf D} +{\bf J}_\lambda] {\bf p}_1 + 
\left[ 2i\varphi' {\bf D}{\bf p}'_0 + i \varphi''{\bf D}{\bf p}_0\right] = 
{\bf D} \left\{[-\varphi'^{2} {\bf I} +{\bf B}_\lambda] {\bf p}_1 + 
\left[ 2i\varphi' {\bf p}'_0 + i \varphi''{\bf p}_0\right]
\right\} . 
$$ 
 
We solve the leading order equations by setting $\varphi '^2$ equal to an eigenvalue
of ${\bf B}_\lambda$ and set ${\bf p}_0 = Q_0(x){\bf p}_*(x)$ where ${\bf p}_*(x)$ is the unit magnitude eigenvector of
$[-\varphi'^{2} {\bf I}  +{\bf B}_\lambda]$
with zero eigenvalue and $Q_0(x)$ is an undetermined scalar function. Then 
$$ - [-\varphi'^{2} {\bf I}  +{\bf B}_\lambda] {\bf p}_1 = \left[ 2i\varphi' {\bf p}'_0 + i \varphi''{\bf p}_0\right] = 
i\left[ 2\varphi'Q'_0 +  \varphi''Q_0\right]{\bf p}_* 
+ 2i\varphi'Q_0{\bf p}'_*.
$$
The matrix premultiplying ${\bf p}_1$ has zero determinant and hence the existence of a solution requires a solvability condition.  

Let ${\bf s}^T_*(x)$ be the zero left eigenvector of unit magnitude of $[-\varphi'^{2} {\bf I}  +{\bf B}_\lambda]$. Then we have the solvability condition
${\bf s}^T_*(x)[-\varphi'^{2} {\bf I}  +{\bf B}_\lambda]={\bf 0}$
by Fredholm's alternative, and thus multiplying by ${\bf s}^T_*(x)$ we have
\beq \label{solv}
\left[ 2\varphi'Q'_0 +  \varphi''Q_0\right]{\bf s_*}^T {\bf p}_* 
+  2\varphi'Q_0{\bf s_*}^T{\bf p}'_*=0,
\eeq 
which yields  
$$ \frac{Q_0'}{Q_0} = -\frac{\varphi''}{2\varphi'}- \frac{{\bf s_*}^T{\bf p}_*'}{{\bf s_*}^T{\bf p_*}}. $$
Thus
\begin{equation}\label{Q0} Q_0(x) = \frac{Q_{00}}{\sqrt{\varphi'}} \exp \left[-\int_a^x \frac{{\bf s_*}(\bar{x})^T{\bf p}'_*(\bar{x})}{{\bf s_*}(\bar{x})^T{\bf p}_*(\bar{x})} \mathrm{d}\bar{x} \right],\end{equation}
where $Q_{00}$ is a constant, not necessarily real, $a$ is a real constant and $\varphi'$ is given by 
either the positive  or the  negative square root of the eigenvalues of ${\bf B}_\lambda$, and $a$ is an arbitrary real constant before any constraints of considering real solutions  and the boundary conditions are imposed. Hence, for each eigenvalue of ${\bf B}_\lambda(x)$, denoted $$\mu^\pm_\lambda(x)\equiv \varphi'^2(x),$$ we have a possible  mode which, at leading order, can be written as
\begin{equation}
  \label{eq:WKBsolu}
{\bf w}(x,t) = e^{\lambda t} \exp \left[-\int_a^x \frac{{\bf s_*}(\bar{x})^T{\bf p}'_*(\bar{x})}{{\bf s_*}(\bar{x})^T{\bf p}_*(\bar{x})} \mathrm{d}\bar{x} \right]
 \frac{1}{[\mu^\pm_\lambda]^{1/4}(x)} \bigg\{C_{0}^{{ \pm}} \cos\left[\frac 1\epsilon  \int_a^x \sqrt{\mu^\pm_\lambda(\bar{x})}\mathrm{d}\bar{x}\right]\\ + S_{0}^{{ \pm}} 
\sin\left[\frac 1\epsilon  \int_a^x \sqrt{\mu^\pm_\lambda(\bar{x})}\mathrm{d}\bar{x}\right]
\bigg\} {\bf p}_*(x),  
\end{equation}
where $C_{0}^{{\pm}},S_{0}^{{\pm}}$ are arbitrary constants. We note that the reciprocal of $\epsilon$ from the trigonometric functions will dominate in spatial derivatives given our asymptotic assumptions.

In the usual Turing analysis, we assume that the steady state is stable to homogeneous perturbations, which are associated with the zero mode. 
{\bl By analogy, 
we therefore seek stability to perturbations which 
 have spatial  derivatives of, at most,  the same O(1) scale as the steady state, in contrast to the WKBJ modes in Eqn. \ref{eq:WKBsolu}, where spatial  derivatives scale with $1/\epsilon$, due to the $1/\epsilon$ factors in the trigonometric contributions. Solutions with the required behaviour are found by setting $\varphi=0$ in the WKBJ expansion or, equivalently by considering
${\bf w}(x,t) = e^{\lambda t} {\bf q}(x)$ 
with a {\it regular} perturbation expansion for ${\bf q}(x)$. One finds ${\bf J}_\lambda(x) {\bf q}(x)={\bf 0}$ at leading order 
and hence the derivatives of ${\bf q}(x)$ are independent of $\epsilon$ and thus order unity at this level of approximation. Hence, for such solutions, substituting ${\bf w}(x,t) = e^{\lambda t} {\bf q}(x)$ 
into Eqn. \ref{weqn1} reduces to}
\beq
{\bf w}_t =  \epsilon^2 {\bf D } {\bf w}_{xx} + {\bf J}(x){\bf w} \approx {\bf J}(x){\bf w},
\eeq
since  $|{\bf w}_{xx}| \sim O(1)$ for this kind of perturbation. Stability of the equilibrium to such perturbations is required for all $x$ and thus, to asymptotic accuracy, we require
\beq \label{tc1i} 
\mbox{tr}({\bf J}(x)) = f_u+g_v < 0, ~~~~ \mbox{det}({\bf J }(x)) = f_ug_v-f_vg_u >0, ~~~~ \text{for all } x\in[0,1],
\eeq
a set of two constraints that we shall assume throughout the text below. These conditions generalize the notion of stability against homogeneous perturbation in the spatially homogeneous setting, and imply that any unstable mode of the form \eqref{eq:WKBsolu} will lead to emergent spatial patterning that is not strictly dictated by the spatial heterogeneity in the kinetics.

The above expression for the leading order solution ${\bf w}$ is well defined with the exception of zeros of $\mu_\lambda^\pm$ (which we will show can be excluded in Proposition \ref{p5} below) and potential singular points $x_*\in[0,1]$ where ${\bf s_*}^T {\bf p_*}=0$. These singular points will in fact determine the subsets of the spatial domain in which a patterning instability will occur.  We first consider properties of $\lambda$ and $\mu_\lambda^\pm$, independently from the solution structure given by \eqref{eq:WKBsolu}, in an arbitrary interval $(a,b) \subset [0,1]$, and then discuss how to choose such an interval so that solutions can be defined. After this, we discuss how solutions behave near these singular points in order to define solutions globally in space. Note, in particular, that we will restrict attention to the open interval $(a,b)$, as we will eventually choose these boundaries to (possibly) be singular points.

\subsection{Local Turing Conditions}

%We consider \eqref{eq:WKBsolu} defined for $x \in (a,b)$, where the end-points are either singular points corresponding to ${\bf s_*}^T {\bf p_*}=0$, or they are the end-points of the full domain. We then enforce either homogeneous Neumann or Dirichlet boundary conditions of the solution. The former will be taken at boundaries where $a=0$ or $b=1$, and the latter used whenever $a$ or $b$ corresponds to a singular point. These boundaries will influence both the choice of the constants $C_{0}$ and $S_{0}$, as well as the properties of $\mu_\lambda^\pm$. We note that both eigenvalues $\mu_\lambda^\pm$ will be continuous functions of $\lambda$, and will be differentiable with respect to the growth rate everywhere that they are distinct \cite{serre2002matrices}. Finally we recall from Section \ref{Overview} that we will think of $\lambda$ as a parameter, and deduce conditions necessary to obtain real solutions when $\lambda$ has a positive real part.

Motivated by the form of solution \eqref{eq:WKBsolu}, we now consider the quantities $\lambda$ and $\mu_\lambda^\pm$ in their own right, and deduce properties that these quantities must have to make sense of such a solution. 
If the WKBJ solution,  \eqref{eq:WKBsolu}, is defined everywhere  on the domain $x \in[0,1]$, Neumann boundary conditions at the domain edges entail, at leading order in $\epsilon$, that  the integral 
\begin{equation}
  \label{eq:FundConstr00001}
    \int_0^1 \sqrt{\mu_{\lambda}^\pm(\bar{x})}\mathrm{d}\bar{x}
\end{equation} 
is a multiple of $\pi \epsilon/2$ and, without loss,  $a=S_0^\pm=0$ so that only cosine solutions remain. Furthermore, given the choice of a sufficiently small $\epsilon$, the above constraint can be ensured simply by imposing 
\begin{equation}
  \label{eq:FundConstr01}
    \int_0^1 \sqrt{\mu_{\lambda}^\pm(\bar{x})}\mathrm{d}\bar{x}\in \mathbb{R}\setminus\{0 \}.
\end{equation} 

However, we are not guaranteed that the WKBJ solution,  \eqref{eq:WKBsolu}, is defined everywhere on the domain $x \in[0,1]$. In general, the region of validity for the WKBJ solution will be restricted to one or more intervals of the form $(a,b)$ where $a$ and $b$ are constants satisfying $0 \leq a < b \leq 1$. The pathological  case of a region of validity that is restricted to a point is neglected,  as this requires mathematical precision in the parameter values. As discussed in Section \ref{Sec.BoundsOfWKB}, when the WKBJ solutions are not valid everywhere, homogeneous Dirichlet boundary conditions will be required (i.e.~whenever $a,b$ satisfy $a>0$ or $b<1$). Given sufficiently small $\epsilon$, as discussed immediately above, then we may accommodate both cases of homogeneous boundary conditions, Neumann or Dirichlet, by requiring (i) without loss a restriction of the WKBJ solution to either a sine or a cosine by setting either $C_0^\pm=0$ or $S_0^\pm=0$ to guarantee the boundary condition at $x=a$ and (ii)
\begin{equation}
  \label{eq:FundConstr}
    \int_a^b \sqrt{\mu_{\lambda}^\pm(\bar{x})}\mathrm{d}\bar{x}\in \mathbb{R}\setminus\{0 \}, 
\end{equation} 
for $0 \leq a < b \leq 1$ so that   the integral in \eqref{eq:FundConstr} can be guaranteed to be a multiple  of either $\pi \epsilon$ or $\pi\epsilon/2$ for a suitable, sufficiently small, choice of $\epsilon$ to guarantee the homogeneous boundary condition at $x=b$.

From the fundamental constraint \eqref{eq:FundConstr}, we can deduce properties of $\lambda$ and $\mu^\pm_\lambda$ that are necessary for a solution to exist within an arbitrary region of validity, $(a,b)$, with $\Re(\lambda) > 0$, therefore resulting in a perturbation that grows in amplitude.
 An analogous derivation in the spatially homogeneous setting is given in SI Section \ref{App.RederivationOfClassDDI}, and is the motivation for what follows. Critically, we assume that conditions \eqref{tc1i} hold for every proposition below. We remark that here we derive implications directly from the fundamental constraint \eqref{eq:FundConstr} that will let us make sense of solutions of the form given in \eqref{eq:WKBsolu}, but do not assume \emph{a priori} that such solutions must be valid.

We first define \emph{permissible} growth rates and eigenvalues which satisfy \eqref{eq:FundConstr}.
\begin{definition} A permissible pair $(\lambda, \mu^\pm_\lambda(x))$ is such that the value of $\lambda$ entails $\mu^\pm_\lambda(x)$ satisfies constraint \eqref{eq:FundConstr} for all $x$ in some interval $(a,b)$. 
\end{definition}
\noindent
We will also refer to $\lambda$ as permissible,  or $\mu^\pm_\lambda(x)$ as permissible, if $(\lambda, \mu^\pm_\lambda(x))$ is permissible, as defined above, and implicitly assume this is over an interval $(a,b)$. 

\begin{proposition} \label{p2}  The function $\mu^\pm_\lambda(x)$ is permissible if and only if  $\mu^\pm_\lambda(x)$ is real and non-negative for all $x\in(a,b)$, though not identically zero. Additionally, without loss of generality we can consider the integral in \eqref{eq:FundConstr} to be positive.
\end{proposition} 
    
\begin{proof} 
If $\mu^\pm_\lambda(x)$ is real, non-negative and not identically zero for $x\in(a,b)$ it is immediate that it is permissible.
For the converse, we consider a permissible $\mu^\pm_\lambda(x)$. 
Note that the square root in condition \eqref{eq:FundConstr} is, without loss of generality, the positive square root.
In other words, we work in the complex plane such that any argument, denoted $\theta$ below, is in the range $\theta \in [0,2\pi)$ and the positive square root is such  that if, for example, $\theta/2 \in [0,\pi)$, then
$$ \left(e^{i\theta}\right)^{1/2} = e^{i\theta/2}.$$ 
Hence, any imaginary  contribution to the 
integrand, 
$\sqrt{\mu^\pm_\lambda(x)}$, 
in condition (\ref{eq:FundConstr}) is non-negative as the argument of the square root in the complex plane is in the range $[0,\pi)$. So any imaginary  contribution to the integrand cannot be cancelled from a contribution in any other region of the integration domain. 
Thus, the integrand cannot have an imaginary contribution and 
 $\sqrt{\mu^\pm_\lambda(x)}$ must be real for all $x\in(a,b)$. Hence, 
$\mu^\pm_\lambda(x)$ is real and non-negative for all $x\in(a,b)$. Finally, the need for the integral to not be identically zero implies that $\mu^\pm_\lambda(x)$ cannot be identically zero.
\end{proof} 

\begin{proposition} \label{p3} If $\lambda$ is both permissible and complex (i.e.~$\Im(\lambda)\neq 0$), then $\Re(\lambda)<0.$
\end{proposition}

\begin{proof} 
%This proceeds analogously to the homogeneous case. 
From the definition of $\mu^\pm_\lambda(x)$, we have 
$$ \mbox{det} [-\mu^\pm_\lambda(x) {\bf D} + {\bf J}_\lambda(x)] =0, $$ 
and, so,
\beq \label{muleqn} 
2 \lambda &=& -\mbox{tr}(\mu^\pm_\lambda(x) {\bf D}-{\bf J})
\pm\sqrt{[\mbox{tr}(\mu^\pm_\lambda(x) {\bf D}-{\bf J})]^2-4 \mbox{det}[\mu^\pm_\lambda(x) {\bf D}-{\bf J}] }
 \nonumber = [f_u+g_v -\mu^\pm_\lambda(x)(1+d)] \\ 
 &\pm& \sqrt{ (f_u+g_v-\mu^\pm_\lambda(x)(1+d))^2 - 4 (d(\mu^\pm_\lambda(x))^2-(df_u+g_v)\mu^\pm_\lambda(x)+(f_ug_v-g_uf_v))}, 
\eeq
with the spatial dependence of $\mu^\pm_\lambda(x)$ such that the growth rate, $\lambda$, does not have a dependence on $x$. Given $\lambda$ is permissible, so that $\mu^\pm_\lambda(x)$ is permissible, we have that $\mu^\pm_\lambda(x)$ is real and non-negative for all $x \in (a,b)$. In addition, tr$({\bf J})=f_u+g_v<0$ for 
all $x$ by Equation \eqref{tc1i} implying that $f_u+g_v -\mu^\pm_\lambda(x)(1+d) <0$ and, thus, if a permissible $\lambda$ is complex it has a negative real part. 
\end{proof}

%We can now show that the Turing conditions given by \eqref{tchet} originate from the requirement of permissible eigenvalues and non-negative growth rates.

\begin{proposition} \label{p4} Given $\Re(\lambda)\geq0$, the pair $(\lambda, \mu^\pm_\lambda(x))$ is permissible if and only if
\beq \label{tc233} \rm{tr}( {\bf B}_{\lambda}) > 0, ~~~~~~ [\rm{tr}( {\bf B}_{\lambda}{)}]^2-4 {\rm det}({\bf B}_{\lambda}) \geq 0,
\eeq 
for all $x \in(a,b)$. 
\end{proposition} 

\begin{proof}
If  $(\lambda, \mu^\pm_\lambda(x))$ is permissible, with $\Re(\lambda)\geq 0$, then $\lambda$ is real by Proposition \ref{p3}. From permissibility and Proposition \ref{p2} we also have $\mu^\pm_\lambda(x)$ is real, non-negative and not identically zero for all $x\in(a,b)$. 
From $ \mbox{det}(-\mu^\pm_\lambda(x) {\bf D} + {\bf J}_\lambda(x)) =0 $, we have 
\beq   \label{ie3} \mu^\pm_\lambda(x) &=& \frac 12 \left[\mbox{tr}( {\bf B}_{\lambda})
\pm\sqrt{[\mbox{tr}( {\bf B}_{\lambda})]^2-4 \mbox{det}({\bf B}_{\lambda}) }\right].
\eeq
As $\mu^\pm_\lambda(x)$ and $\lambda$ are strictly real, this enforces
$$ [\rm{tr}( {\bf B}_{\lambda})]^2-4 {\rm det}({\bf B}_{\lambda}) \geq 0, $$ 
for all $x\in(a,b)$. 
We also have det$({\bf B}_{\lambda})= {\rm det}({\bf D}^{-1}{\bf J_\lambda})={\rm det}({\bf D}^{-1})~{\rm det}({\bf J_\lambda}) = {\rm det}({\bf J_\lambda})/d >0$ for all $x\in(a,b)$, using Equation \eqref{tc1i} and that $\lambda$ is real and non-negative. Hence, for both the positive and negative square root in Equation \eqref{ie3}, the fact that $\mu^\pm_\lambda(x)$ cannot be negative enforces $\rm{tr}( {\bf B}_{\lambda}) \geq 0$ for all $x\in(a,b)$. The possibility that $\rm{tr}( {\bf B}_{\lambda}) = 0$ is excluded as then $\mu^\pm_\lambda(x)$ is not real, since  det$({\bf B}_{\lambda})>0.$

Conversely, assuming conditions \eqref{tc233}, we can see by Equation \eqref{ie3} that $\mu_\lambda^\pm(x) > 0$ for all $x \in (a,b)$, and, hence, condition \eqref{eq:FundConstr} is satisfied. \end{proof}

We note that conditions \eqref{tc233} cannot be satisfied in the case $d=1$, as then $\rm{tr}( {\bf B}_{\lambda}) = \rm{tr}( {\bf J})-2\lambda < \rm{tr}({\bf J}) < 0$, so for any permissible $\lambda$ with $\Re(\lambda)>0$, we must have $d < 1$. As the conditions in \eqref{tc233} do not depend on the positive or negative branch of $\mu_\lambda^\pm$, we immediately have that this proposition implies both roots are permissible once one of them is. We also need the following Proposition, which shows that if these conditions hold for $\lambda > 0$, then they hold for $\lambda = 0$.

\begin{proposition} \label{newp} If the conditions \eqref{tc233} hold for some real permissible $\lambda^*$ with $\lambda^*>0$ for all $x \in (a,b)$, then they hold for all real $\lambda$ with $\lambda^* \geq \lambda \geq 0$ for all $x \in (a,b)$.
\end{proposition} 

\begin{proof}
 First, we realise that  ${\rm tr}({\bf B}_\lambda)={\rm tr}({\bf B}_0)-\lambda(1+1/d)$, and the last term is strictly negative so that from ${\rm tr}({\bf B}_{\lambda^*}) > 0$ we have  ${\rm tr}({\bf B}_\lambda)>0$ for all $0\leq\lambda\leq \lambda^*$. Next we consider the second condition of \eqref{tc233}, which can be written as  
 \begin{equation}
     P(\lambda) := \left (1-\frac{1}{d}\right)^2\lambda^2-\frac{2(d-1)(df_u-g_v)}{d^2}\lambda+\left(f_u+\frac{g_v}{d}\right)^2-4\frac{f_ug_v-f_vg_u}{d} \geq 0,
 \end{equation}
  where the quadratic  $P(\lambda)$ admits two zeros, which we can compute as
 $$
 \tilde{\lambda}^\pm = \frac{g_v-df_u\pm2\sqrt{-df_vg_u}}{1-d}.
 $$
 We can see that these roots are both real by signing each term. From ${\rm tr} ({\bf J})<0$ while ${\rm tr}({\bf B}_{\lambda^*})={\rm tr}({\bf D}^{-1}{\bf J}_{\lambda^*})>0$ for $\lambda^*>0$ we have that $f_u+ g_v<0$ while $f_u+g_v/d>0$ and $d< 1$. Hence $f_u<0,~g_v>0$ and $f_u-g_v/d<0$.
 
 Due to the positive coefficient of $\lambda^2$ in $P(\lambda)$, we can see that this function must have a negative minimum, so that it is positive to the left of $\tilde{\lambda}^-$ and to the right of $\tilde{\lambda}^+$ (and negative between these roots). If $\lambda^* \leq \tilde{\lambda}^-$, then $P(\lambda)$ is decreasing in $\lambda$ and, so,  $P(\lambda) >0$ for all $0\leq\lambda\leq \lambda^*$, meaning that we would be done. Hence, we now assume that $\lambda^* \geq \tilde{\lambda}^+$ to derive a contradiction. By linearity we have $0 < {\rm tr}({\bf B}_{\lambda^*}) \leq {\rm tr}({\bf B}_{\tilde{\lambda}^+})$. We then compute,
 $$
 {\rm tr}({\bf B}_{\tilde{\lambda}^+}) = \frac{-(2d+2)\sqrt{-df_vg_u}+2d(f_u-g_v)}{(1-d)d}<0,
 $$
which can be seen as the denominator is strictly positive and the numerator has only negative terms. Therefore we must have $\lambda^* \leq \tilde{\lambda}^-$, so that $P(\lambda) \geq 0$ for all $\lambda < \lambda^*$.
 \end{proof}

Next we show a relationship between the positive and negative eigenvalues of ${\bf B}_\lambda$, and how they depend on $\lambda$. We will need to assume that $[\rm{tr}({\bf B}_{\lambda})]^2-4 {\rm det}({\bf B}_{\lambda}) > 0$ for all $x \in (a,b)$. If this term becomes zero (and hence $\mu_\lambda^-(x)=\mu_\lambda^+(x)$), then there is a degeneracy in the associated eigenvectors, which will lead to an internal boundary-layer behaviour discussed in the next section, and hence the determination of the boundary points, $a$ and $b$.

\begin{proposition}\label{p5}
Given $(\lambda, \mu^\pm_\lambda(x))$ is permissible, $\Re(\lambda)\geq0$, and $[\rm{tr}( {\bf B}_{\lambda})]^2-4 {\rm det}({\bf B}_{\lambda}) > 0$ for all $x \in (a,b)$, we then have the ordering $0 < \mu_0^-(x) < \mu^-_\lambda(x) < \mu^+_\lambda(x) < \mu_0^+(x)$ for all $x \in (a,b)$. Furthermore, at the edges of the domain, $x=a$, or $x=b$, we still have the ordering $0 < \mu_0^-(x) \leq \mu^-_\lambda(x) \leq \mu^+_\lambda(x) \leq \mu_0^+(x)$.
\end{proposition}

\begin{proof}
Using $\partial_\lambda$ to denote differentiation with respect to $\lambda$, we have that $\partial_\lambda {\rm tr}({\bf B}_{\lambda})=-{\rm tr}({\bf D}^{-1})=-(1+1/d)$ and $\partial_\lambda {\rm det}({\bf B}_{\lambda})=-{\rm tr}({\bf J}_\lambda)/d=(2\lambda -f_u - g_v)/d$. Then, by Equation \eqref{ie3}, $(\mu_\lambda^\pm(x))^2-{\rm tr}({\bf B}_{\lambda})\mu_\lambda^\pm(x) + {\rm det}({\bf B}_{\lambda}) = 0$, so, upon taking the derivative and rearranging we have,
%\beq   \label{ie3new} \partial_\lambda \mu^\pm_\lambda(x) &=& \frac 12 \left[-\mbox{tr}( {\bf D}^{-1})
%\pm\frac{-\mbox{tr}( {\bf D}^{-1}{\bf J_\lambda})\mbox{tr}( {\bf D}^{-1})+2{\rm tr}({\bf J}_\lambda)/d}{\sqrt{[\mbox{tr}( {\bf D}^{-1}{\bf J_\lambda})]^2-4 \mbox{det}({\bf D}^{-1}{\bf J_\lambda}) }}\right].
%\eeq
%We see that the first term of \eqref{ie3new} is always negative, while the second will only change sign based on the numerator. By Proposition \ref{p4}, and the fact that ${\rm tr}({\bf J}_\lambda) < {\rm tr}({\bf J}_0) <0$, we have that this numerator is always negative for $\mu_\lambda^+$, and hence $\mu_\lambda^+ < \mu_0^+$.
\beq   \label{ie3new} \partial_\lambda \mu_\lambda^\pm(x)=\frac{-d{\rm tr}({\bf D}^{-1})\mu_\lambda^\pm(x)+{\rm tr}( {\bf J}_\lambda)}{d(2\mu_\lambda^\pm(x) -{\rm tr}({{\bf B}_\lambda}))}=\frac{-d{\rm tr}({\bf D}^{-1})\mu_\lambda^\pm(x)+{\rm tr}( {\bf J}_\lambda)}{\pm d\sqrt{[\mbox{tr}({{\bf B}_\lambda})]^2-4 \mbox{det}({{\bf B}_\lambda})}}.
\eeq
We can then see that each term in the numerator is always negative for both roots, whereas the denominator will change sign. Hence, we have $\partial_\lambda \mu_\lambda^-(x) > 0$ and $\partial_\lambda \mu_\lambda^+(x) < 0$, so that the ordering follows by continuity. Finally, we note that the possibility of $\mu_0^- \leq 0$ is excluded using Equation \eqref{ie3} along with Proposition \ref{p2} and ${\rm tr}({\bf B}_{0})>0$, the latter of which is true by virtue of Proposition \ref{newp}. 

For the second part, all of the ordering can be deduced as a limit of the above argument (with the new potential equality $\mu_\lambda^-=\mu_\lambda^+$ if $[\rm{tr}( {\bf B}_{0})]^2-4 {\rm det}({\bf B}_{0}) = 0$ by Equation \eqref{ie3}) except the definite inequality $\mu_0^->0. $   To rule out $\mu_0^-=0$, we consider $\lambda=0.$ If $[\rm{tr}( {\bf B}_{0})]^2-4 {\rm det}({\bf B}_{0}) > 0$ at the boundary point as well, the above proof holds.
Hence we need now only consider the case with  $[\rm{tr}( {\bf B}_{0})]^2-4 \rm{det}({\bf B}_{0}) =0.$ Now  suppose, for contradiction, that $\mu_0^-=0$. By \eqref{ie3} we have
$\rm{tr}({\bf B}_0)=0$ and hence 
${\rm det}({\bf B}_{0}) =\mbox{det}({\bf J})/d =0,$ but we have  $\mbox{det}({\bf J})>0$ throughout, and hence the contradiction.\end{proof}

%For the second part, if  $[\rm{tr}( {\bf B}_{\lambda}(x))]^2-4 {\rm det}({\bf B}_{\lambda}(x)) \neq 0$ at some $x \in (a,b)$, then the reasoning above follows identically. Assume then that for some $x^* \in (a,b)$ that $[\rm{tr}( {\bf B}_{\lambda}(x))]^2-4 {\rm det}({\bf B}_{\lambda}(x)) = 0$. Then we have $\mu_\lambda^-(x^*)=\mu_\lambda^+(x^*)$ from Equation \eqref{ie3}. As this zero is simple, then for $0 \leq \tilde{\lambda} <\lambda$ one can show that $\partial_{\tilde{\lambda}}\left[({\rm tr}({\bf B}_{\tilde{\lambda}}(x))^2-4{\rm det}({\bf B}_{\tilde{\lambda}}(x))\right]<0$ (see Equation \eqref{nmon}), so we can define $\partial_{\tilde{\lambda}} \mu_{\tilde{\lambda}}^\pm(x^*)$ by Equation \eqref{ie3new}, and obtain the same ordering as before except for $\mu_\lambda^-(x^*)=\mu_\lambda^+(x^*)$.

%Note that the positive branch should generate a larger set of WKB solutions than the negative one as $0 < \mu_\lambda^- < \mu_\lambda^+$ at every point of $\mathcal{T}_\lambda$ and hence \eqref{hetselect} is satisfied for larger values of n for the larger eigenvalue than the smaller one. 
Propositions \ref{newp} and \ref{p5} together give a range of permissible values of $\lambda$ and associated eigenvalues $\mu_\lambda^\pm$, as soon as the conditions \eqref{tc233} are satisfied for some positive $\lambda^*>0$. Finally, we show that for a permissible $\lambda$ with $\Re(\lambda)\geq0$, and the same assumption as above, we can sensibly define the left and right eigenvectors ${\bf s}_*$ and ${\bf p}_*$ which are not orthogonal.

\begin{proposition}\label{nonorth}
Given $(\lambda, \mu^\pm_\lambda(x))$ is permissible and $\Re(\lambda)\geq 0$, then $[\rm{tr}( {\bf B}_{\lambda})]^2-4 {\rm det}({\bf B}_{\lambda}) > 0$ for all $x \in (a,b)$ if and only if ${\bf s}_*^T {\bf p}_*\neq0$ for all $x \in (a,b)$, where ${\bf s}_*$ and ${\bf p}_*$ are the left and right unit eigenvectors of $[-\mu_\lambda^\pm{\bf I}+ {\bf B}_\lambda]$. 
\end{proposition}
\begin{proof}
We will demonstrate both implications via contraposition. We first assume that ${\bf s}_*^T {\bf p}_*=0$ at some point $x_* \in (a,b)$. By elaborating possibilities on a case by case basis for a general $2\times 2$ matrix with zero determinant, we note the left and right eigenvectors of zero eigenvalue can only be perpendicular if the matrix is proportional to one of the following:
$$ \left( \begin{array}{cc} 0 & 0 \\ 0 & 0 \end{array} \right) , 
~~~~~~~
\left( \begin{array}{rr} 1 & 1 \\ -1 & -1 \end{array} \right) ,~~~~~~~
\left( \begin{array}{rr} 1 & -1 \\ 1 & -1 \end{array} \right) . 
$$

In all three cases, we have that the trace is zero.
Therefore, $\rm{tr}(-\mu_\lambda^\pm{\bf I}+ {\bf B}_\lambda) = -2\mu_\lambda^\pm+\rm{tr}({\bf B}_\lambda)=0.$ However, by Equation \eqref{ie3}, this implies that $[\rm{tr}( {\bf B}_{\lambda})]^2-4 {\rm det}({\bf B}_{\lambda}) = 0$, contradicting the assumption that this quantity remains positive.

For the converse, we assume that $[\rm{tr}( {\bf B}_{\lambda})]^2-4 {\rm det}({\bf B}_{\lambda}) = 0$ at some point $x_*\in (a,b)$ (noting that if this term were negative, then, by Proposition \ref{p4}, $\lambda$ would not be permissible and we would have an immediate contradiction). By using Equation \eqref{ie3} again we see that $\rm{tr}(-\mu_\lambda^\pm{\bf I}+ {\bf B}_\lambda)=0,$ so this matrix then has repeated zero eigenvalues. Any real $2\times 2$ matrix with zero determinant and trace can be written as either,
$$ \left( \begin{array}{cc} c_1 & c_2 \\ -\frac{c_1^2}{c_2} & -c_1 \end{array} \right) , ~~~\text{or}~~~~
\left( \begin{array}{cc} 0 & 0 \\ 0 & 0 \end{array} \right) , 
$$
for real $c_1, c_2$, with $c_2 \neq 0$. The first of these has one left and one right eigenvector, given by ${\bf s}_* = (c_1/c_2,1)$ and ${\bf p}_* = (-c_2/c_1,1)$, which are orthogonal. The second of these would imply $f_v=g_u=0$, and, along with the assumption that ${\rm det}({\bf J})>0$, we would have $f_ug_v > 0$, so these terms must have the same sign. But noting that ${\rm tr}({\bf J})<0$ and ${\rm tr}({\bf B}_\lambda)>0$, by assumption on the stability of the zero mode, permissibility of $\lambda$, and Proposition \ref{p4}, we have $f_u + g_v < 0$ and $f_u + g_v/d - \lambda/d>0$, thus, we see that they must have opposite signs, demonstrating that this case is not possible. Therefore, if $[\rm{tr}( {\bf B}_{\lambda})]^2-4 {\rm det}({\bf B}_{\lambda}) = 0$ at some point $x_*\in (a,b)$ for permissible $\lambda$, then ${\bf s}^T_*{\bf p}_*=0$ at this point.
\end{proof}

Given Propositions \ref{p4}, \ref{p5} and \ref{nonorth}, which all follow from the definition of permissible growth rates, we can now consider where  solutions of the form given in Equation \eqref{eq:WKBsolu} for permissible $\lambda \geq0$ are valid. We will assume throughout that  $[{\rm tr}({\bf B}_\lambda)]^2-4 {\rm det}({\bf B}_\lambda)$ only has simple zeros, noting that non-simple zeros would require mathematical fine tuning of parameters. When $[{\rm tr}({\bf B}_\lambda)]^2-4 {\rm det}({\bf B}_\lambda)\geq0$, we have by  the first part of Proposition \ref{p5} that the reciprocal of $[\mu_\lambda^\pm]^{1/4}$ with permissible $\lambda\geq0$ is nonsingular and thus Equation \eqref{eq:WKBsolu} with permissible $\lambda\geq0$ might only possess a singularity at points where left and right eigenvectors are orthogonal, that is ${\bf s}_*^T {\bf p}_*=0$. Then, for a region with 
 ${\rm tr}({\bf B}_\lambda) >0$, and $[{\rm tr}({\bf B}_\lambda)]^2-4 {\rm det}({\bf B}_\lambda)\geq 0$ we have that $\lambda$ is permissible by Proposition \ref{p4} and 
 that the reciprocal of $[\mu_\lambda^\pm]^{1/4}$  is  nonsingular, even at  the domain edges by the additional use of the second part of 
  Proposition \ref{p5}. We  define  the closure of the maximal open set where
the associated WKBJ solution for this $\lambda$ is non-singular 
  by $\mathcal{T}_\lambda$. By the above reasoning and Proposition \ref{nonorth} each boundary of this region must either be a domain boundary, or where $[{\rm tr}({\bf B}_\lambda)]^2-4 {\rm det}({\bf B}_\lambda)=0$ as ${\bf s}_*^T {\bf p}_*\neq 0$ on the interior of $\mathcal{T}_\lambda$. Whenever the latter case occurs,  $\mathcal{T}_\lambda \neq [0,1]$ and we have to determine what happens to the WKBJ solution on approaching the point where ${\bf s}_*^T {\bf p}_*\neq 0$ and beyond.

\subsection{Behaviour Near Singular Points}\label{Sec.BoundsOfWKB}

If $\mathcal{T}_\lambda=[0,1]$, for a given permissible  $\lambda\geq 0$, then we can take $a=S_0^\pm=0$ in \eqref{eq:WKBsolu} to find a non-trivial solution that satisfies the homogeneous Neumann boundary conditions. If instead $\mathcal{T}_\lambda$ is a proper subset of  the whole interval $[0,1]$ then we assume for simplicity that $\mathcal{T}_\lambda$ is a single contiguous interval, implying that $[{\rm tr}({\bf B}_\lambda)]^2-4 {\rm det}({\bf B}_\lambda)$ has at most two zeros for $x \in [0,1]$, and note that generalizing beyond a single interval is straightforward. At a zero of $[{\rm tr}({\bf B}_\lambda)]^2-4 {\rm det}({\bf B}_\lambda)$, by \eqref{ie3} we have the double eigenvalues, $\mu_\lambda^-=\mu_\lambda^+$, of ${\bf B}_\lambda$ and we recap that by Proposition \ref{nonorth} we have ${\bf s}_*^T {\bf p}_*=0$ at such a point, denoted $x_*$, and thus anticipate a singularity in the solution given by Equation \eqref{eq:WKBsolu}. In SI Section \ref{Sec.PropsOfWKB}, we explicitly show that for   $y>x>x_*$ the integral
\beq \label{integral} \exp \left[ \int_x^{y} \frac{{\bf s_*}(\bar{x})\cdot{\bf p}_*'(\bar{x})}{{\bf s_*}(\bar{x})\cdot{\bf p}_*(\bar{x})} \mathrm{d}\bar{x} \right],
\eeq
will blow up {\bl as $x \searrow x_*$,} with analogous behaviour when approaching such a singular point from the left. However, in this SI section, we also show that this integral will scale such that by imposing effective Dirichlet conditions at the singular point, we can retain a bounded solution. In this way we can construct leading order solutions which are bounded and defined on $\mathcal{T}_\lambda$. 

 With the previously stated assumption that  any zero of $[{\rm tr}({\bf B}_\lambda)]^2-4 {\rm det}({\bf B}_\lambda)$ is simple, so that $[{\rm tr}({\bf B}_\lambda)]^2-4 {\rm det}({\bf B}_\lambda)$ monotonically passes through zero at such a singular point and by Proposition \ref{nonorth} we thus have that $[\rm{tr}( {\bf B}_{\lambda})]^2-4 {\rm det}({\bf B}_{\lambda}) < 0$ outside of $\mathcal{T}_\lambda$. By Proposition \eqref{p4}, this implies that this value of $\lambda$ is not permissible outside of this interval, and hence if any WKBJ solutions exist, they  cannot simultaneously satisfy homogeneous Dirichlet or Neumann conditions at boundaries on both the left and right. However, the scaling of the integral \eqref{integral} requires a solution which is zero at the singular point, while a zero derivative is always required at a domain boundary. Thus the only WKBJ solution outside of $\mathcal{T}_\lambda$ associated with the growth rate $\lambda$ is the zero solution. We can then extend the nontrivial WKBJ solution defined on $\mathcal{T}_\lambda$ by the zero solution to obtain a leading order solution across the whole domain for a mode with fixed growth rate $\lambda$.

% At the boundary (or boundaries) of $\mathcal{T}_\lambda$ where 
% $[{\rm tr}({\bf B}_\lambda)]^2-4 {\rm det}({\bf B}_\lambda)=0$}

% By Propositions \ref{p3} and \ref{p4}, the only trigonometric solution of the form \eqref{eq:WKBsolu} with $\Re(\lambda)>0$ that can satisfy homogeneous Neumann or Dirichlet conditions is the zero solution, as we would otherwise have $\sqrt{\mu_\lambda^\pm}\in\mathbb{C}\setminus\mathbb{R}$ by Proposition \ref{p4}. At this boundary, there is a potential blow-up in solutions given by \eqref{eq:WKBsolu}, and hence we refer to the boundary as a singular point. To accommodate this singular point, we must understand how the solution scales near such a singularity, and how to choose the arbitrary constants $C_0^\pm$ and $S_0^\pm$t.  

We can now match these different sets of boundary conditions depending on the number of singular points appearing in the domain. These will then lead to different wave number selection conditions. We note in particular that in matching Neumann boundaries, we only differentiate the trigonometric functions in \eqref{eq:WKBsolu}, as only terms involving these derivatives will be retained to leading order. We then have the following leading order solutions (modes) associated with each eigenvalue $\mu_\lambda^\pm$ of ${\bf B}_\lambda = {\bf D}^{-1}{\bf J}_\lambda(x)$, depending on the number of singular points:
\begin{itemize}
\item no singular points, so $\mathcal{T}_\lambda=[0,1]$ and the solution is
  \begin{subequations} \label{eq:PiecewiseWKBsolu}
\begin{equation}
  \label{eq:PiecewiseWKBsolu_full}
{\bf w}^\pm(x,t) = e^{\lambda t} \exp \left[-\int_0^x \frac{{\bf s_*}(\bar{x})^T{\bf p}_*'(\bar{x})}{{\bf s_*}(\bar{x})^T{\bf p}_*(\bar{x})} \mathrm{d}\bar{x} \right]
 \frac{C_{0}^{{\pm}}}{[\mu_\lambda^\pm(x)]^{1/4}} 
\cos\left(\frac 1\epsilon  \int_0^x \sqrt{\mu_\lambda^\pm(\bar{x})}\mathrm{d}\bar{x}\right) 
{\bf p}_*(x), ~~~~ ~~~~~
\int_0^1 \sqrt{\mu_\lambda^\pm(\bar{x})}\mathrm{d}\bar{x} = n^\pm\pi\epsilon;
\end{equation}
\item one singular point $x_*(\lambda)>0$, so without loss of generality, $(x_*,1)= \mathcal{T}_\lambda$, with solution
  \begin{equation}
    \label{eq:PiecewiseWKBsolu_right}
    {\bf w}^\pm(x,t) =  e^{\lambda t} \exp \left[ \int_x^{1} \frac{{\bf s_*}(\bar{x})^T{\bf p}_*'(\bar{x})}{{\bf s_*}(\bar{x})^T{\bf p}_*(\bar{x})} \mathrm{d}\bar{x} \right] \frac{S_{0}^{{ \pm}}}{[\mu_{\lambda}^\pm(x)]^{1/4}} 
\sin\left(\frac 1\epsilon  \int_{x_*}^x \sqrt{\mu_{\lambda }^\pm(\bar{x})}\mathrm{d}\bar{x}\right) 
{\bf p}_*(x), 
 \int_{x_*}^1 \sqrt{\mu_{\lambda }^\pm(\bar{x})}\mathrm{d}\bar{x} = \left(n^\pm+\frac{1}{2}\right)\pi\epsilon,
\end{equation}
for $x \in \mathcal{T}_\lambda$, and zero otherwise;
\item
  two singular points $x_*(\lambda),~x_{**}(\lambda)\in(0,1)$ delimiting the $\mathcal{T}_\lambda$ set, i.e. $\mathcal{T}_\lambda=(x_*,~x_{**})$, with solution
  \begin{equation}
    \label{eq:PiecewiseWKBsolu_interval}
    {\bf w}^\pm(x,t) = e^{\lambda t} \exp \left[ \int_x^{x_{**}} \frac{{\bf s_*}(\bar{x})^T{\bf p}_*'(\bar{x})}{{\bf s_*}(\bar{x})^T{\bf p}_*(\bar{x})} \mathrm{d}\bar{x} \right]
 \frac{S_{0}^{{\pm}}}{[\mu_{\lambda}^\pm(x)]^{1/4}} 
\sin\left(\frac 1\epsilon  \int_x^{x_{**}} \sqrt{\mu_{\lambda}^\pm(\bar{x})}\mathrm{d}\bar{x}\right) 
{\bf p}_*(x), 
 \int_{x_*(\lambda)}^{x_{**}(\lambda)} \sqrt{\mu_{\lambda }^\pm(\bar{x})}\mathrm{d}\bar{x} = n^\pm \pi\epsilon,
\end{equation}
for $x \in \mathcal{T}_\lambda$, and zero otherwise;
\end{subequations}
\end{itemize}
where $C_{0}^{{\pm}},~S_{0}^{{\pm}}$ are arbitrary real constants.  We remark that the mode selection constraint is defined over $\mathcal{T}_\lambda$, and so will depend on $\lambda$ through both the eigenvalues $\mu_{\lambda}^\pm$ and any singularities, as highlighted in the integral bounds. In this way, the latter two solutions given by Equations \eqref{eq:PiecewiseWKBsolu_right}-\eqref{eq:PiecewiseWKBsolu_interval} are continuously extended by zero outside of $\mathcal{T}_\lambda$, and equal to zero at the singular points.

We remark that these WKBJ solutions applied to systems without heterogeneity in the reaction kinetics collapse in both components to functions of the form
$$e^{\lambda t} (C_0 \cos(n \pi x) + S_0 \sin(n \pi x)).$$ However, the meaning of $n$ (denoted as $n^\pm$) in the heterogeneous case does not correspond to the spatial frequency of a given mode, as we will see in an example. We can now describe some additional structural features of the spaces $\mathcal{T}_\lambda$ and how they change for different growth rates $\lambda$. In particular, the set $\mathcal{T}_{\lambda}$ is monotonic in $\lambda$ in the following sense.

\begin{proposition} \label{c1} 
 If $\mathcal{T}_{\lambda_2}\neq\emptyset$ and $0\leq \lambda_1\leq\lambda_2$ then $\mathcal{T}_{\lambda_2}\subseteq \mathcal{T}_{\lambda_1}$. If $\mathcal{T}_{\lambda_1}\neq [0,1]$, and ${0\leq}\lambda_1<\lambda_2$, then we have the stricter inclusion $\mathcal{T}_{\lambda_2}\subset \mathcal{T}_{\lambda_1}$.
\end{proposition}
\begin{proof}
The first part of this for $0\leq \lambda_1\leq\lambda_2$ follows from Proposition \ref{newp}. We then need to show that if $\lambda_1 < \lambda_2$, then $\mathcal{T}_{\lambda_1}\not\subseteq \mathcal{T}_{\lambda_2}$. We note that at least one of the boundaries of $\mathcal{T}_\lambda$, $a(\lambda)$ and/or $b(\lambda)$, are zeros (in the spatial variable $x$) of ${\rm tr}({\bf B}_\lambda(x))^2-4{\rm det}({\bf B}_\lambda(x))$. At such a boundary we compute the derivative with respect to $\lambda$, finding
\begin{equation}\label{nmon}
    \partial_\lambda \left[({\rm tr}({\bf B}_\lambda(x))^2-4{\rm det}({\bf B}_\lambda(x))\right] = %-2{\rm tr}({\bf D}^-1)+4{\rm tr}({\bf J}_\lambda) < 0,
    { -2 {\rm tr}({\bf B}_\lambda) {\rm tr}({\bf D}^{-1})+4 {\rm tr}({\bf J}_\lambda) \det ({\bf D}^{-1})<0,}
\end{equation}
which follows by signing each term. As ${\rm tr}({\bf B}_\lambda(x))^2-4{\rm det}({\bf B}_\lambda(x)) > 0$ for $a(\lambda) < x < b(\lambda)$, we have that if $a(\lambda_1) > 0$ then $a(\lambda_1) < a(\lambda_2)$ and if $b(\lambda_1) < 1$ then $b(\lambda_1) > b(\lambda_2)$, so the strict inclusion $\mathcal{T}_{\lambda_2}\subset \mathcal{T}_{\lambda_1}$ follows.
\end{proof}

Hence, the onset of instability (the boundary of $\mathcal{T}_0$) is given by zeros of $[{\rm tr}({\bf B}_0)]^2-4 {\rm det}({\bf B}_0)=0$, for which ${\rm tr}({\bf B}_0)>0$. More generally, the onset of instability with a growth rate $\lambda\geq0$ is given by the location of zeros of ${\bf s}_*^T {\bf p}_*=0$, i.e. zeros of $[{\rm tr}({\bf B}_\lambda)]^2-4 {\rm det}({\bf B}_\lambda)=0$  while ${\rm tr}({\bf B}_\lambda>0$). Therefore this boundary shifts with $\lambda$, while monotonicity of the Turing space with respect to $\lambda$ holds. Hence, the sufficient condition for (the onset of) instability can be identified with $\mathcal{T}_0$, which is a good approximation of where we will find Turing patterns, as it corresponds to the region of support of a mode with positive value of $\lambda$ for sufficiently small $\epsilon$. This also justifies considering the fundamental constraint \eqref{eq:FundConstr} in this regime, which does not depend on $\epsilon$. 

%In summary: if there is an interval $(a,b)$ where the local Turing conditions are satisfied,
%\begin{equation}
%  \label{eq:localDDI}
%  \mbox{tr}\left({\bf J}(x)\right) < 0, ~~~ \mbox{det}\left({\bf J }(x)\right) >0, ~~~{\rm tr}\left({{\bf D}^{-1}{\bf J}(x)}\right)>0, ~~~ [{\rm tr}\left({{\bf D}^{-1}{\bf J}(x)}\right)]^2-4 {\rm det}\left({{\bf D}^{-1}{\bf J}(x)}\right)>0,\quad \forall x\in(a,b),
%\end{equation}
%then there exists a small enough $\epsilon$ where heterogeneous perturbations grow (more precisely one can identify whether the $n=1$ mode with the smallest support fits in the ($a,b)$ interval \emph{a posteriori}, though the region $\mathcal{T}_0$ given by \eqref{eq:localDDI} is typically sufficient if it is large enough relative to $\epsilon$).

In summary, we have the following conditions for instability:
\begin{theo}[$\lambda$-Dependent Heterogeneous Case]\label{heteroturing} Let $\lambda > 0$, $0<\epsilon \ll 1$, and assume that the quantity {\bl $[{\rm tr}( {\bf B}_\lambda(x))]^2-4 {\rm det}({\bf B}_\lambda(x))$} has no more than two simple zeros for $x \in [0,1]$, and is positive between these two zeros. If we assume stability to perturbations in the absence of diffusion, i.e.,
\beq {\rm tr}( {\bf J}(x)) < 0, ~~ {\rm det}({\bf J}(x)) > 0, ~~~ \text{for all }x \in [0,1], \label{hstin}
\eeq
then there exists a non-homogeneous perturbation ${\bf w}$ satisfying \eqref{weqn1} (to leading order in $\epsilon$) which grows as $e^{\lambda t}$ in the interval $x \in (a(\lambda),b(\lambda))$ if
\beq \label{tchet}  {\rm tr}( {\bf B}_\lambda(x)) > 0, ~~ [{\rm tr}( {\bf B}_\lambda(x))]^2-4 {\rm det}({\bf B}_\lambda(x)) > 0, ~~~ \text{for all }x \in (a(\lambda),b(\lambda)), \eeq
and if there exists an integer $n^\pm>0$ such that
\beq\label{hetselect} \int_{a(\lambda)}^{b(\lambda)} \sqrt{\mu_{\lambda}^\pm(\bar{x})}\mathrm{d}\bar{x} = \left(n^{\pm}+ \frac{K}{2}\right)\pi \epsilon,
\eeq
where $a(\lambda) = \max(0, \min(\{x: [{\rm tr}( {\bf B}_\lambda(x))]^2-4 {\rm det}({\bf B}_\lambda(x)) =0\}))$, $b(\lambda) = \min(1, \max(\{x: [{\rm tr}( {\bf B}_\lambda(x))]^2-4 {\rm det}({\bf B}_\lambda(x)) =0\}))$ and $K=0$ if either $a(\lambda) = 0$ and $b(\lambda)=1,$ or if  $0 <a(\lambda) <b(\lambda)<1;$ otherwise $K=1$. 
\end{theo}
\begin{proof}
We assume without loss of generality that $(a(\lambda),b(\lambda))$ has one of the forms given in \eqref{eq:PiecewiseWKBsolu}. By Proposition \ref{p3} we have no loss in specialising to strictly real $\lambda$. Assuming conditions \eqref{tchet} are satisfied, Propositions \ref{p2} and \ref{p4} imply that $\mu_\lambda^\pm$ is permissible, real, and positive. From this and Proposition \ref{p9}, we have that the functions given by \eqref{eq:PiecewiseWKBsolu} are real and bounded for all $x \in (a(\lambda),b(\lambda))$. To leading order in $\epsilon$, such solutions satisfy \eqref{weqn1}, alongside the zero solution. By the scaling arguments in SI Section \ref{Sec.PropsOfWKB}, we can see that the solutions given by \eqref{eq:PiecewiseWKBsolu} meet this zero solution at any internal boundary (i.e.~any zero of {\bl $[{\rm tr}( {\bf B}_\lambda(x))]^2-4 {\rm det}({\bf B}_\lambda(x))$} in the interval $(0,1)$). So to leading order, such a piecewise solution satisfies $\eqref{weqn1}$ and the Neumann boundary conditions at $\{0,1\}$.
\end{proof}

 Analogous criteria for the other possibilities for $\mathcal{T}_\lambda$, depending on the sign pattern of {\bl $[{\rm tr}( {\bf B}_\lambda(x))]^2-4 {\rm det}({\bf B}_\lambda(x))$} across the domain, are readily determined. Further we note that the integers $n^\pm$ play an analogous role to the wave number $n$ in the homogeneous setting, but that they will not correspond to spatial frequency, and the two roots will have quantitatively different properties, so must be considered as distinct. For sufficiently small $\epsilon$, these conditions predict that a pattern will form in the interval $(a,b)$, and intervals for which no value of $\lambda$ exists will return to the heterogeneous steady state ${\bf u^*}$ after a small perturbation (up to leading order in $\epsilon$). We will confirm this numerically in Section \ref{Illustration}. Additionally, the fact that unstable modes do not share the same support is shown explicitly in Proposition \ref{c1}, and employed to explain some properties of patterns in heterogeneous domains. 
 
We remark that \eqref{hetselect} depends on a given $\lambda$ both in the integrand and the bounds of the integral, but in principle for a given $\epsilon$ and $n^{\pm}$, one can use this condition to find at most two values of $\lambda$ indicating an instability, one for each eigenvalue. Hence, any instability will permit a discrete number of unstable modes, each with a possibly different support, and the growth rate of any instability will, thus, depend locally on the permissible growth rates. We give further structural details regarding $n^\pm$ and $\lambda$ in SI Section \ref{App.Lambdan}.

Further, it should be noted that Criterion \ref{heteroturing} can be generalized to obtain Criterion \ref{hettheo} by relaxing the restriction to a single interval, and considering a suitable choice of arbitrarily small $\epsilon$. The use of the interior of $\mathcal{T}_0$ in this limit for Criterion \ref{hettheo}  is further supported by Proposition \ref{c1}, and noting that instabilities need not grow on the edges of $\mathcal{T}_0$ for the WKBJ solutions at leading order (in particular this is the case when the homogeneous Dirichlet boundary conditions are imposed to retain bounded solutions there).
{\blue However, although linked, we highlight that Criterion \ref{hettheo}  and  \ref{heteroturing} are different. Criterion  \ref{heteroturing} gives conditions  for the presence a specific unstable WKBJ Turing mode, which is subject to the wavemode selection relation of Eqn.~(\ref{hetselect}). Aside from a trivial relaxation of the requirement that the WKBJ Turing mode may only have support  within a single interval, Criterion \ref{hettheo} gives conditions that ensure that there is at least one destabilising WKBJ  Turing mode with a positive growth rate for sufficiently small $\epsilon$, effectively by amalgamating Criterion \ref{heteroturing} across all possible modes. Hence Criterion \ref{hettheo} is the most useful in classifying whether or not there is a Turing instability. Nonetheless, Criterion \ref{heteroturing} is a fundamental stepping  stone to deriving Criterion \ref{hettheo} and, in addition,  provides detailed information, for example  about the location of the support of the WKBJ Turing mode solutions and the   relation between  the growth rate to the non-dimensional diffusion coefficient, via Eqn.~(\ref{hetselect}).  }

\section{Illustrative example. The Schnakenberg model}\label{Illustration}
To illustrate our results, we consider the Schnakenberg model with spatially heterogeneous sources. Let 
$$ {\bf u} = \left( \begin{array}{c} u_1 \\ u_2 \end{array} \right), $$ 
so that $u_1$ is the nominal inhibitor and $u_2$ is the nominal activator. The kinetics are
$$ {\bf F}({\bf u},x)= \left( \begin{array}{c} \beta(x) -u^2_2u_1 \\ u^2_2u_1-\alpha u_2+\zeta(x) \end{array} \right), $$
with $\alpha,~ \beta(x), ~\zeta(x) >0$. 
As is typical, and to simplify the system, we assume $\beta(x)+\zeta(x)=1$. Hence,  accurate to $O(\epsilon^2)$, the steady state is given by 
% $$ {\bf u}^* = \left( \begin{array}{c}  \dfrac{\alpha^2 \beta(x)}{(\beta(x)+\zeta(x))^2} \\ ~~  \\ \dfrac 1 \alpha(\beta(x)+\zeta(x)) \end{array} \right),  $$
 $$ {\bf u}^* = 
 \left( \begin{array}{c}
\alpha^2 \beta(x) \\ ~~
 \\ \dfrac 1 \alpha
 \end{array} \right), 
 $$
 {\bl Note that ${\bf u}^*(x)$ has no dependence on $\epsilon$ at this order of approximation by construction and furthermore 
 $\beta(x)$ is taken so that  ${\bf u}^* (x)$  
 satisfies the boundary conditions at the domain edges, as required in the derivation of  Eqn.~\eqref{hetss} and discussed in detail 
 in Section \ref{Overview}.  We proceed by evaluating the Jacobian, which is 
 given by}
 %$$ {\bf J} = \left( \begin{array}{c|c}
 %-\dfrac 1{ \alpha^2} (\beta(x)+\zeta(x))^2 & -2\alpha \dfrac{ \beta(x)}{\beta(x)+\zeta(x)}    \\ ~~ &  ~~ 
 %\\ \dfrac 1{ \alpha^2} (\beta(x)+\zeta(x))^2 & 
 %\alpha\left(  \dfrac{\beta(x)-\zeta(x)}{\beta(x)+\zeta(x)}  \right) 
 %\end{array} \right), 
 %$$ 
  $$ {\bf J} = \left(
  \begin{array}{c|c}
 -\dfrac 1{ \alpha^2} & -2\alpha \beta(x)
 \\ ~~ &  ~~ 
 \\ \dfrac 1{ \alpha^2}  & 
 \alpha\left( 2 \beta(x)-1  \right) 
 \end{array} 
 \right), 
 $$ 
 so that 
 \beq \label{locturcond}
 {\rm tr}({\bf J}) = \alpha \left( 2 \beta(x)-1  \right)  - \dfrac 1{ \alpha^2}, ~~~ {\rm det}({\bf J}) = \dfrac 1 \alpha , ~~~
 {\rm tr}({\bf D}^{-1} {\bf J}) =  \dfrac \alpha d\left( 2 \beta(x)-1  \right)  - \dfrac 1{ \alpha^2}, \\ \nonumber  [{\rm tr} ({\bf D}^{-1} {\bf J} )]^2 - 4 {\rm det}({\bf D}^{-1} {\bf J}) = \dfrac{(2\beta(x)-1)^2\alpha^6+(-4\beta(x)-2)d \alpha^3+d^2}{\alpha^4 d^2}. 
 \eeq
 {\bl We first note that the second constraint of the Turing condition, Eqn. \ref{hstin}, is automatically satisfied since  ${\rm det}({\bf J})>1.$ } 	
 	  To satisfy the Turing conditions \eqref{tchet} for $\lambda=0$ we require 
  \begin{subequations} \label{schnackc1}
\begin{equation}
  \label{schnackc1a}
  4\alpha^6\beta(x)^2-4(\alpha^6+\alpha^3 d)\beta(x)+\alpha^6-2\alpha^3 d+d^2>0,
  \end{equation} \begin{equation}\label{schnackc1b}
  \beta(x) > \dfrac 12 \left( 1+ \dfrac d{\alpha^3}  \right) ,  ~~~~  \beta(x) < \dfrac 12 \left( 1+ \dfrac 1{\alpha^3}\right).
  \end{equation}
  \end{subequations}
Thus, inequalities \eqref{schnackc1b} require $d<1$, as standard. Condition \eqref{schnackc1a} forces $\beta(x)$ to lie outside of the roots of this quadratic, i.e.,
$$ \beta(x) > \dfrac 12 \left( 1+ \dfrac d{\alpha^3}  \right)+\sqrt{\dfrac d{\alpha^3}}, ~~ \text{or} ~~\beta(x) < \dfrac 12 \left( 1+ \dfrac d{\alpha^3}  \right)-\sqrt{\dfrac d{\alpha^3}}. $$
The second of these inequalities cannot be reconciled with the first inequality of \eqref{schnackc1b}. We then have the conditions on the parameters for a Turing instability {\bl (at the marginal case of $\lambda=0$)} are that $\alpha >0$, $0 < d < 1$, and for all $x \in \mathcal{T}_0$, 
$$  \dfrac 12 \left( 1+ \dfrac d{\alpha^3}\right)+\sqrt{\dfrac d{\alpha^3}} < \beta(x) < \dfrac 12 \left( 1+ \dfrac 1{\alpha^3}  \right),$$
{\bl with the latter inequality also enforcing the 
first  constraint of the Turing condition, Eqn.~\eqref{hstin}.
The above inequality is also accompanied by the need to ensure  an unstable mode satisfies condition \eqref{hetselect}, analogous to the \emph{a posteriori} selection of a wave number for the spatially homogeneous Turing instability. This condition will be satisfied if and when the boundary conditions are satisfied.}

 % These conditions, which really are just local versions of the standard Turing conditions, can be modified to determine if any value of $\lambda>0$ permits an unstable mode.

\subsection{Direct Numerical Solutions}
We simulated system \eqref{RDE} with the Schnakenberg kinetics. Initial data were taken as normally distributed spatial perturbations to ${\bf u^*}$. Specifically, we set $u_i(0) = u_i^*(1+\xi_i(x))$ where $i=1,2$ and $\xi_i(x) \sim \mathcal{N}(0, 10^{-3})$ independently for each $x$ and $i$. While such heterogeneous reaction-diffusion systems are standard problems for numerical simulation software, we carefully checked different implementations of our simulations in order to be sure we resolved boundary layers and solution structure in the spatial domain. The commercial finite-element solver Comsol version 5.4 was used to solve the equations with $10^5$ elements, a relative tolerance of $10^{-4}$, and a final time of $t=10^6$ by which time a steady state had been reached up to numerical tolerances. Simulations were also carried out using a standard three-point stencil in Matlab and the stiff solver \textit{ode15s}, using $10^4$ grid points with relative and absolute tolerances of $10^{-9}$, and the same solutions were found. WKBJ modes were reconstructed in Mathematica and these were checked in Matlab and Maple. 

\begin{figure}
    \centering
    \includegraphics[width=0.45\textwidth]{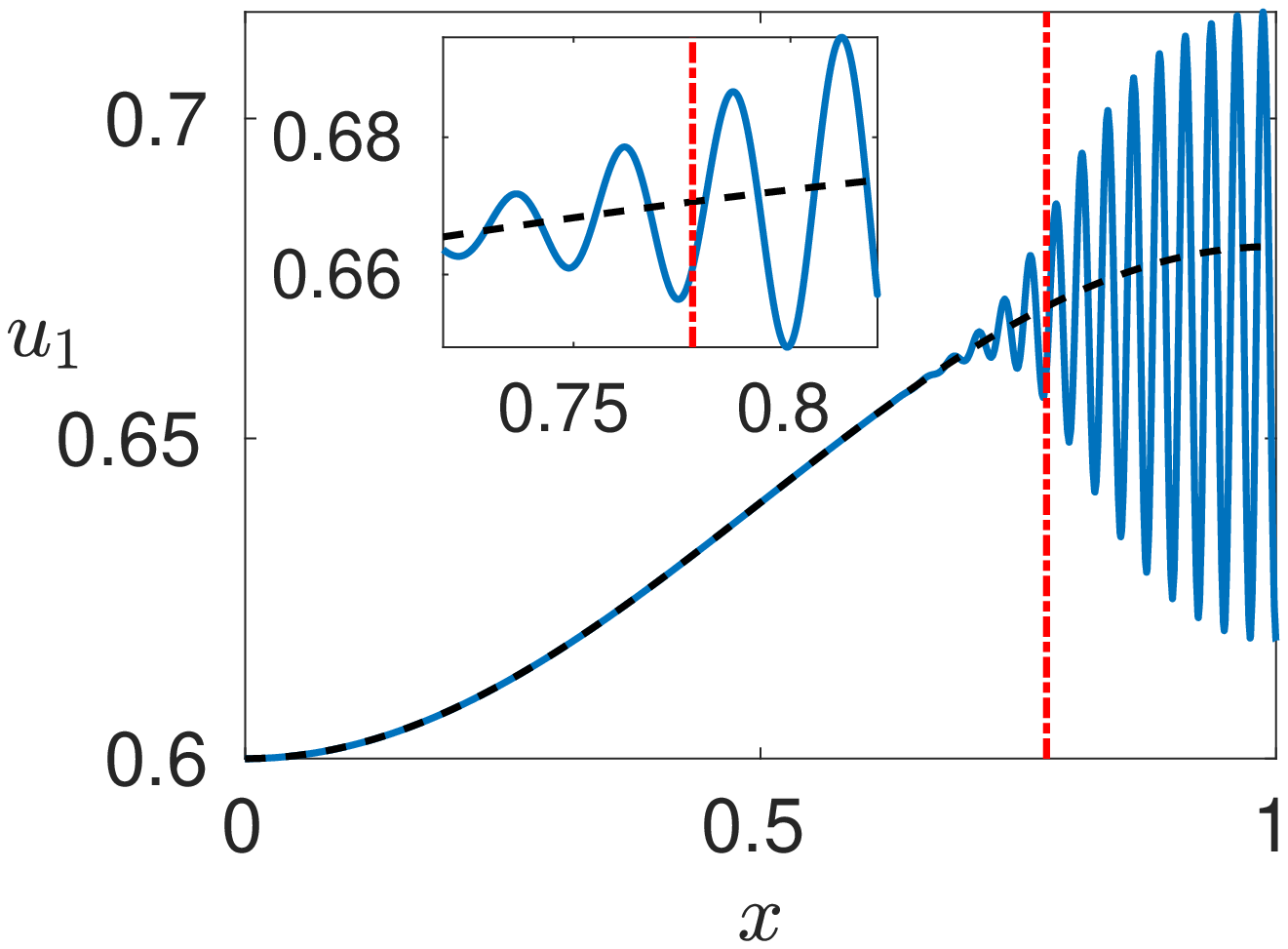}
    \includegraphics[width=0.45\textwidth]{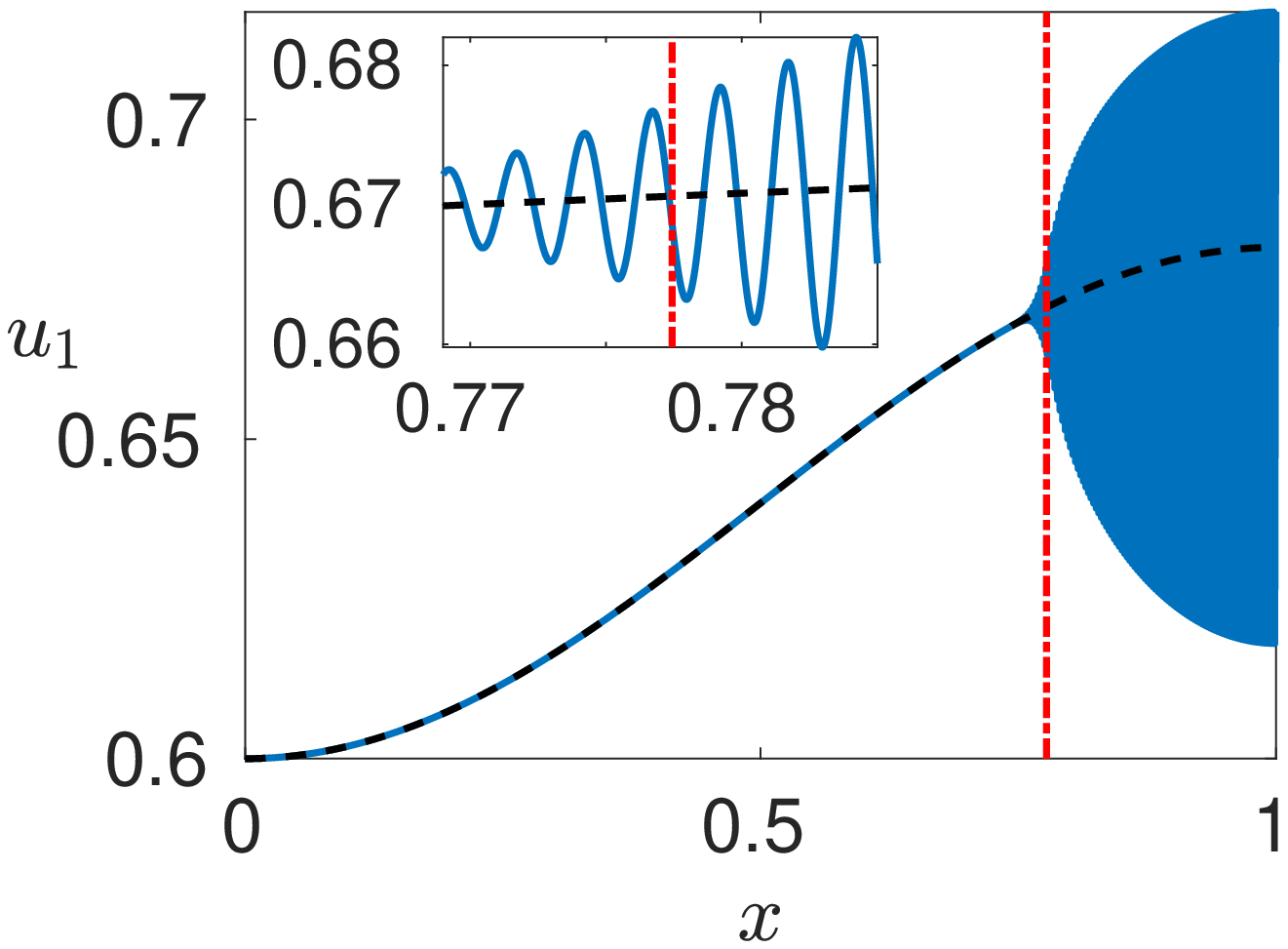}
    
    \hspace{1cm} (a) $\epsilon = 10^{-2}$ \hspace{5cm} (b) $\epsilon = 10^{-3}$
    
    \includegraphics[width=0.45\textwidth]{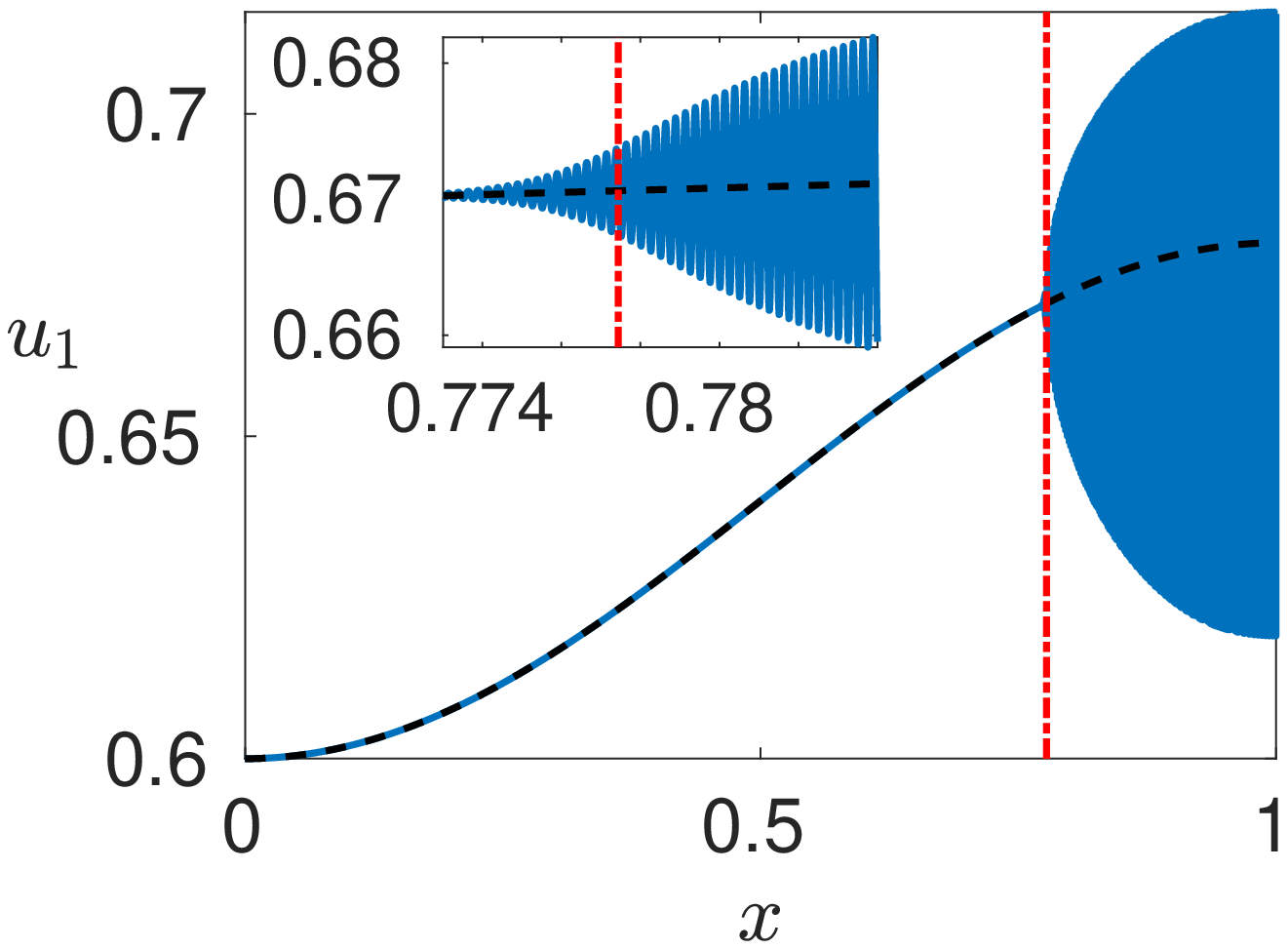}
    \includegraphics[width=0.45\textwidth]{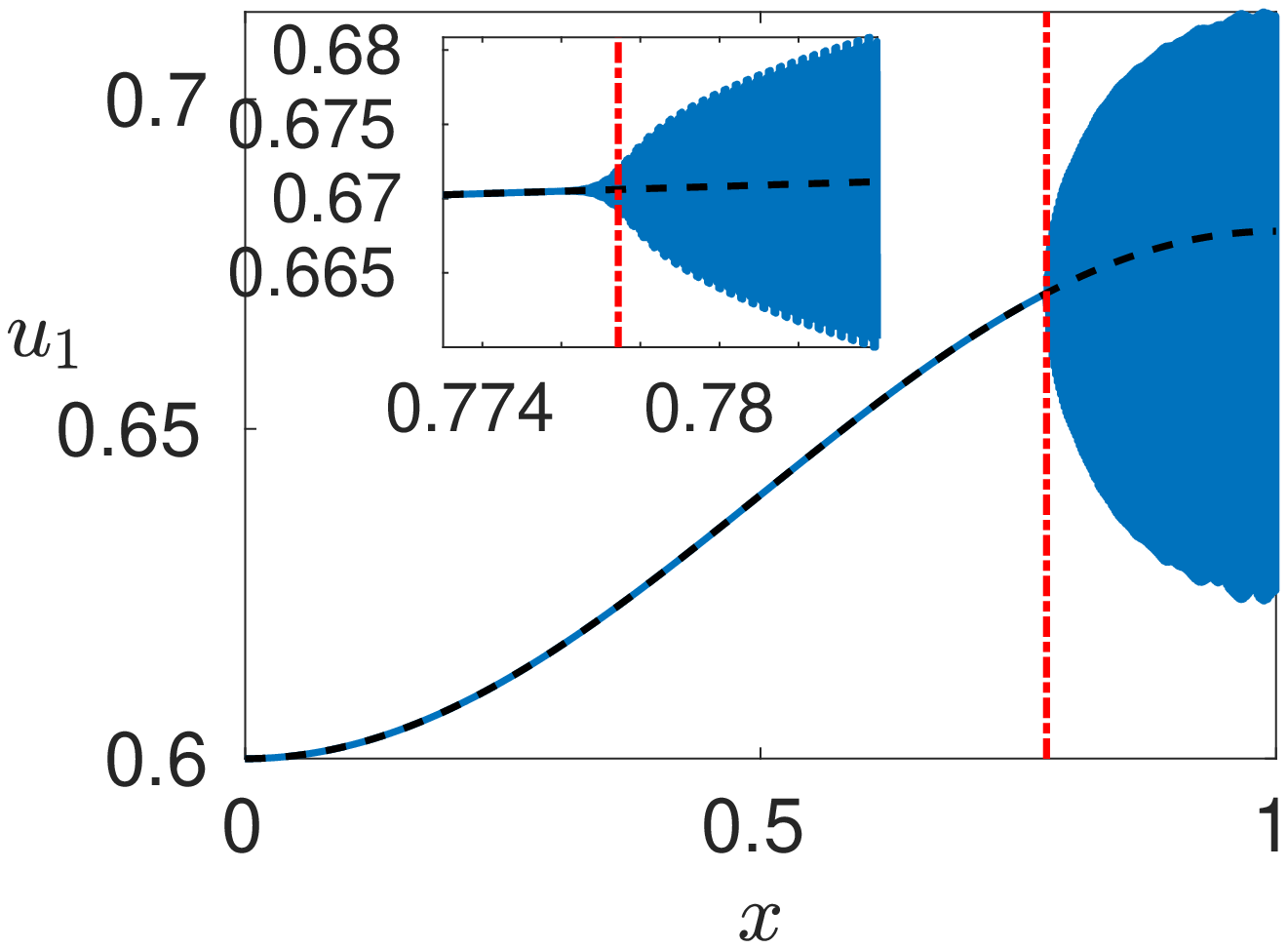}
    
\hspace{1cm} (c) $\epsilon = 10^{-4}$ \hspace{5cm} (d) $\epsilon = 10^{-5}$

    \caption{Plots of $u_1$ from simulations of the Schnakenberg system using $\alpha=1$, $d=1/40=0.025$, and $\beta=3/5+[1-\cos(\pi x)]/25$ with varying $\epsilon$. The blue solid curve is from the numerical simulation whereas the black dashed curve is the stationary state $u_1^*(x) = \beta(x)$ (note that the blue \emph{region} is due to the highly oscillatory nature of the solution). The red dash-dotted curve is the boundary of $\mathcal{T}_0$ at $x \approx 0.7774$ (i.e.~the singular point $x_*$). The insets show a zoomed-in region near the boundary of $\mathcal{T}_0$; these insets are over different regions in (a) and (b), though Figs.~(c) and (d) share the same $x$ axis for their insets.}
    \label{fig1}
\end{figure}

\begin{figure}
    \centering
    \includegraphics[width=0.45\textwidth]{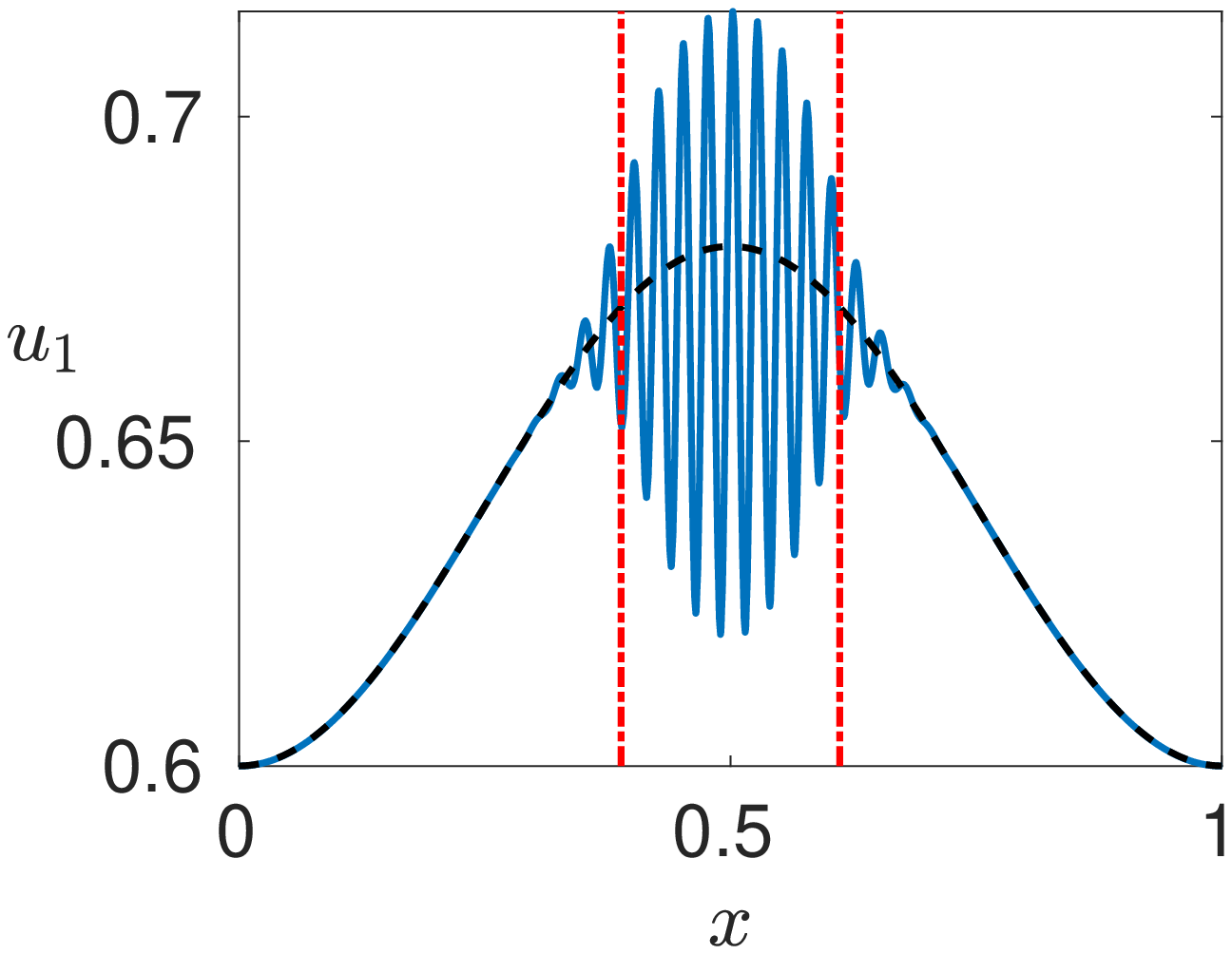}
    \includegraphics[width=0.45\textwidth]{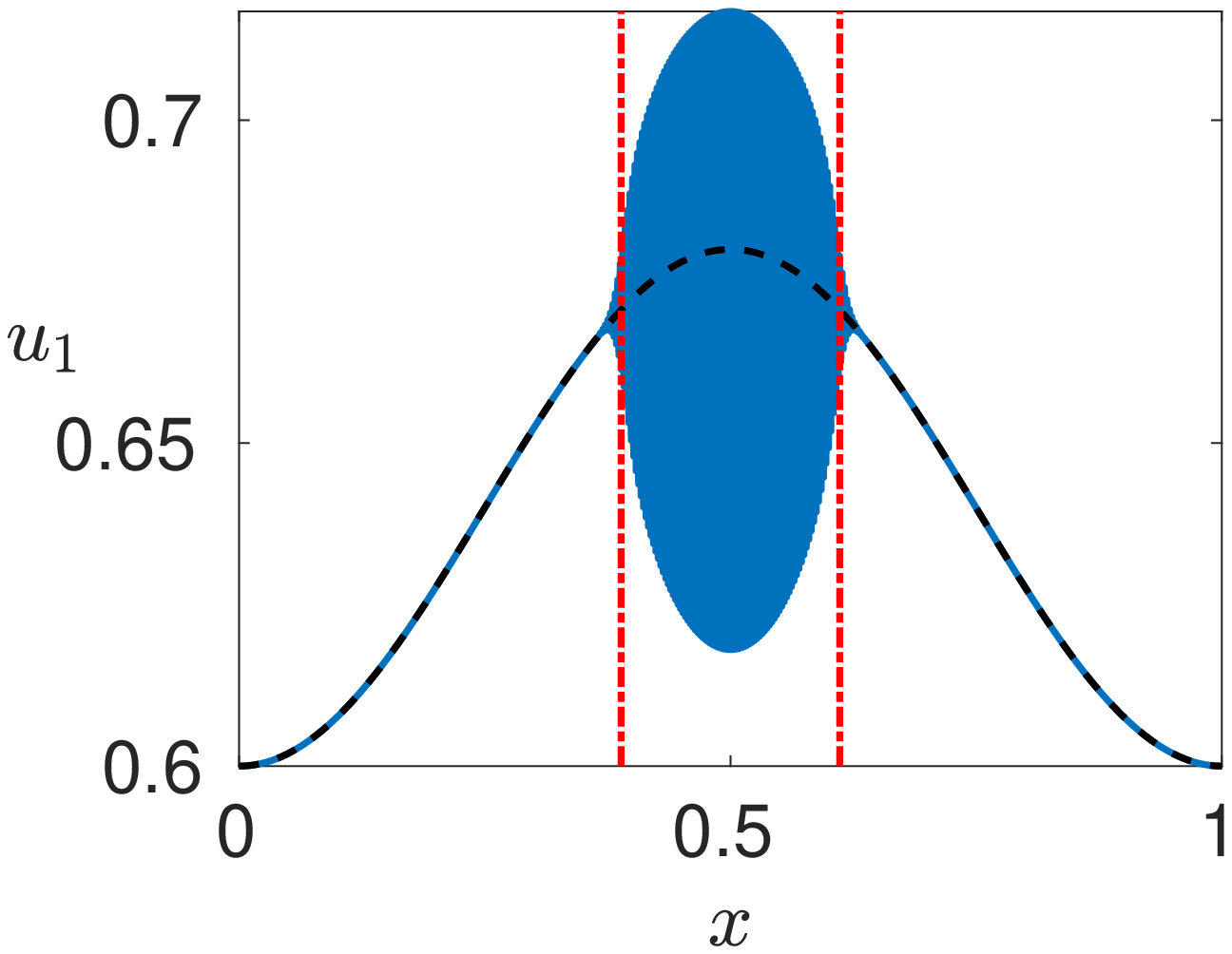}
    
    \hspace{1cm} (a) $\epsilon = 10^{-2}$ \hspace{5cm} (b) $\epsilon = 10^{-3}$

    \caption{Plots of $u_1$ as in Fig.~\ref{fig1}, but with $\beta=3/5+[1-\cos(2\pi x)]/25$ to demonstrate an internally contained $\mathcal{T}_0$.}
    \label{fig1b}
\end{figure}

We demonstrate our results using the following parameter choices, unless otherwise stated. We take $\alpha=1$, $d=1/40=0.025$, and consider $\beta=3/5+[1-\cos(c\pi x)]/25$, where $c=1$ or $c=2$. For these parameters, we have a Turing instability if $0.6706 < \beta(x) < 1$, so for $c=1$ we have $\mathcal{T}_0 \approx (0.7774,1)$ and for $c=2$ we have $\mathcal{T}_0 \approx (0.3886, 0.6114)$. We plot simulations for $c=1$ in Fig.~\ref{fig1}, and vary $\epsilon$. We observe that patterned solutions form approximately in the region predicted by the analysis, $\mathcal{T}_0$, and that they localize to this region as $\epsilon$ is decreased with highly oscillatory boundary regions at $x_* \approx 0.7774$. We note that Figs.~\ref{fig1}(b)-(d) have the same qualitative structure in terms of the amplitudes of patterns, though the internal oscillations become increasingly finer as $\epsilon$ is decreased. The insets show the increasing localization of the boundary as $\epsilon$ is decreased, as well as the structure of the decaying boundary layer of the mode with the largest support. We also show the same kind of localization for $c=2$ in Fig.~\ref{fig1b} where the spike solutions are confined to an internal region by varying the heterogeneity. Larger values of $c$, as well as other kinds of heterogeneity, were also considered with results consistent with the analytical predictions.

\begin{figure}
    \centering
    \includegraphics[width=0.45\textwidth]{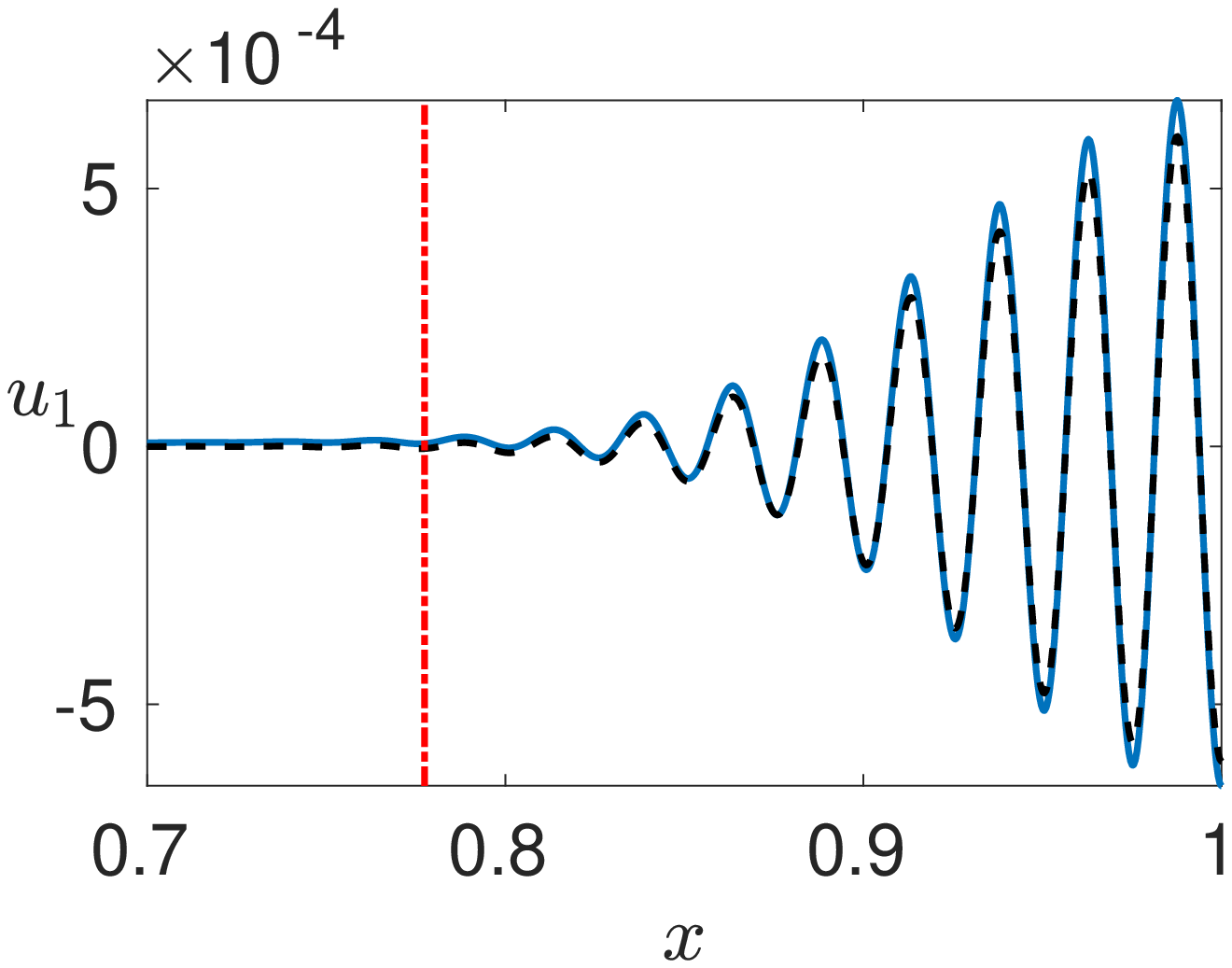}
    \includegraphics[width=0.45\textwidth]{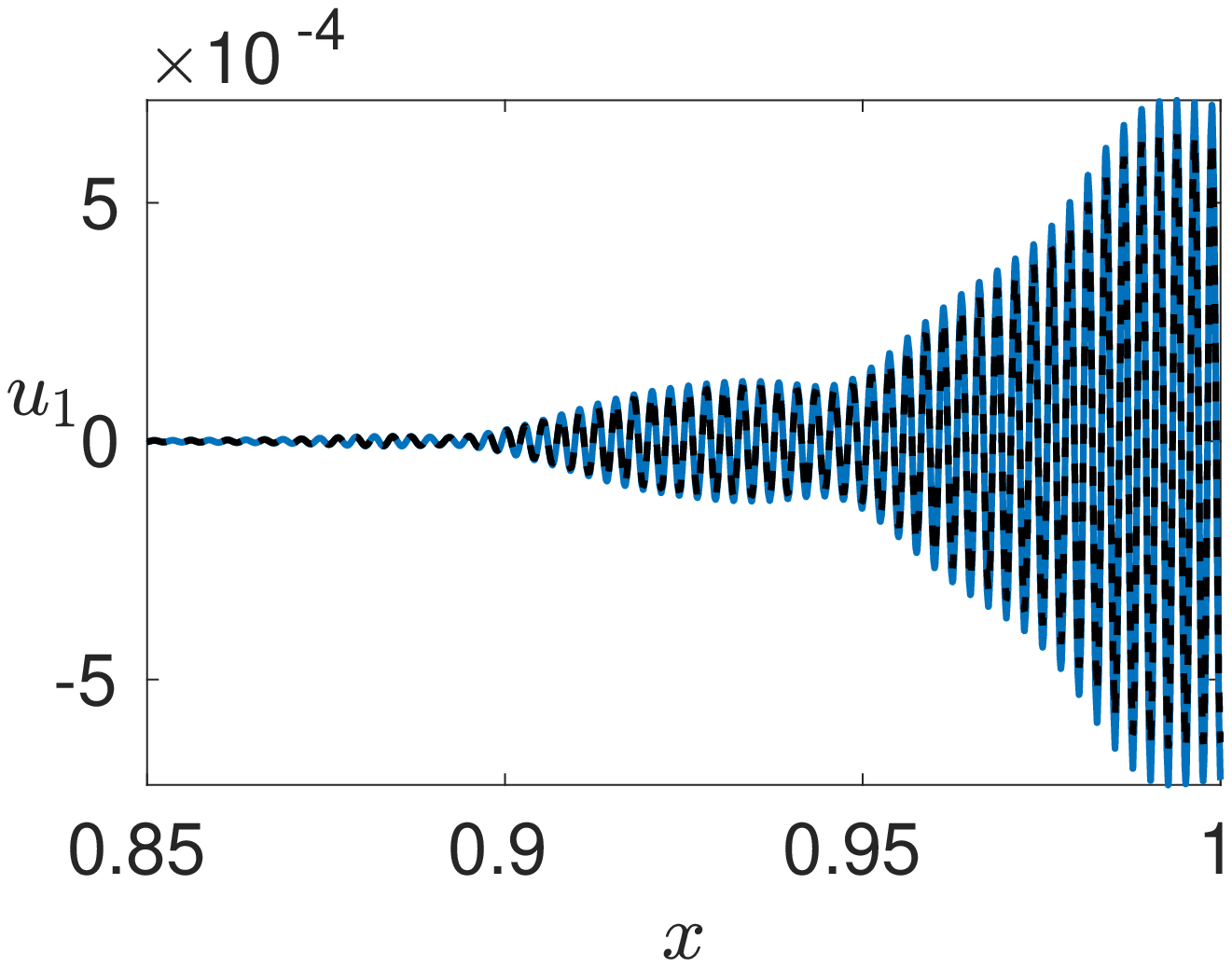}
    
    \hspace{1cm} (a) $\epsilon = 10^{-2}$ \hspace{5cm} (b) $\epsilon = 10^{-3}$

    \caption{Plots of $u_1$ from simulations of the Schnakenberg system using $\alpha=1$, $d=1/40=0.025$, and $\beta=3/5+[1-\cos(\pi x)]/25$ with varying $\epsilon$ at a times $t_f=800$ in (a) and $t_f=700$ in (b). The initial perturbation is taken as $\xi_i(x) \sim \mathcal{N}(0, 10^{-6})$ in both cases. The blue solid curve is given by $u_1(t_f)-u_1^*$ from the full numerical simulation, and the black dashed curve is given by $w_1(t_f)$ from simulations of the linearized system. The red dash-dotted curve is the boundary of $\mathcal{T}_0$ at $x \approx 0.7774$, though the region shown in (b) is entirely within $\mathcal{T}_0$.}
    \label{fig2}
\end{figure}

In Fig.~\ref{fig2} we show short time solutions to the nonlinear and linearized system in order to understand the structure of growing modes due to the instability. As anticipated, for sufficiently small perturbations and time intervals, the linear and nonlinear simulations are almost identical (using the same realization of the initial perturbations). We observe that the instability grows fastest furthest to the right, suggesting that there is not a single largest growth rate $\lambda$ across the domain, as anticipated in the analysis. Rather, what we have plotted are a superposition of modes with distinct growth rates and supports. Finally, for smaller $\epsilon$, these results suggest that larger values of $\lambda$ (which are more localized) become permissible, which is consistent with the structures anticipated. We now explore these modes in more detail.

\subsection{Structure of Unstable Modes}
  
We construct the unstable modes given by \eqref{eq:PiecewiseWKBsolu_right} for the example shown in Fig.~\ref{fig1}(a) with $\epsilon=0.01$, and discuss their properties. First, we numerically determine the discrete plausible mode numbers $n^\pm\in\mathbb{N}$ and $\lambda$ from the constraint \eqref{hetselect}. Then ${\bf B}_\lambda$, ${\bf p}_*, {\bf s}_*, \mu_\lambda$ and ${\bf w}$ all follow from their definitions. This example is indicative of the general features of linearly unstable modes in heterogeneous reaction-diffusion systems; the restriction to an example with modes of the form given by \eqref{eq:PiecewiseWKBsolu_right} is just for clarity of presentation, and our qualitative observations generalize. Specifically, unstable regions $\mathcal{T}_\lambda$ which are composed of many disjoint intervals will in general have a wide variety of unstable modes across the domain, but the analysis in any such complicated setting will essentially reduce to the structures found here.

    \begin{figure}
    \centering
    \includegraphics[width=0.45\textwidth]{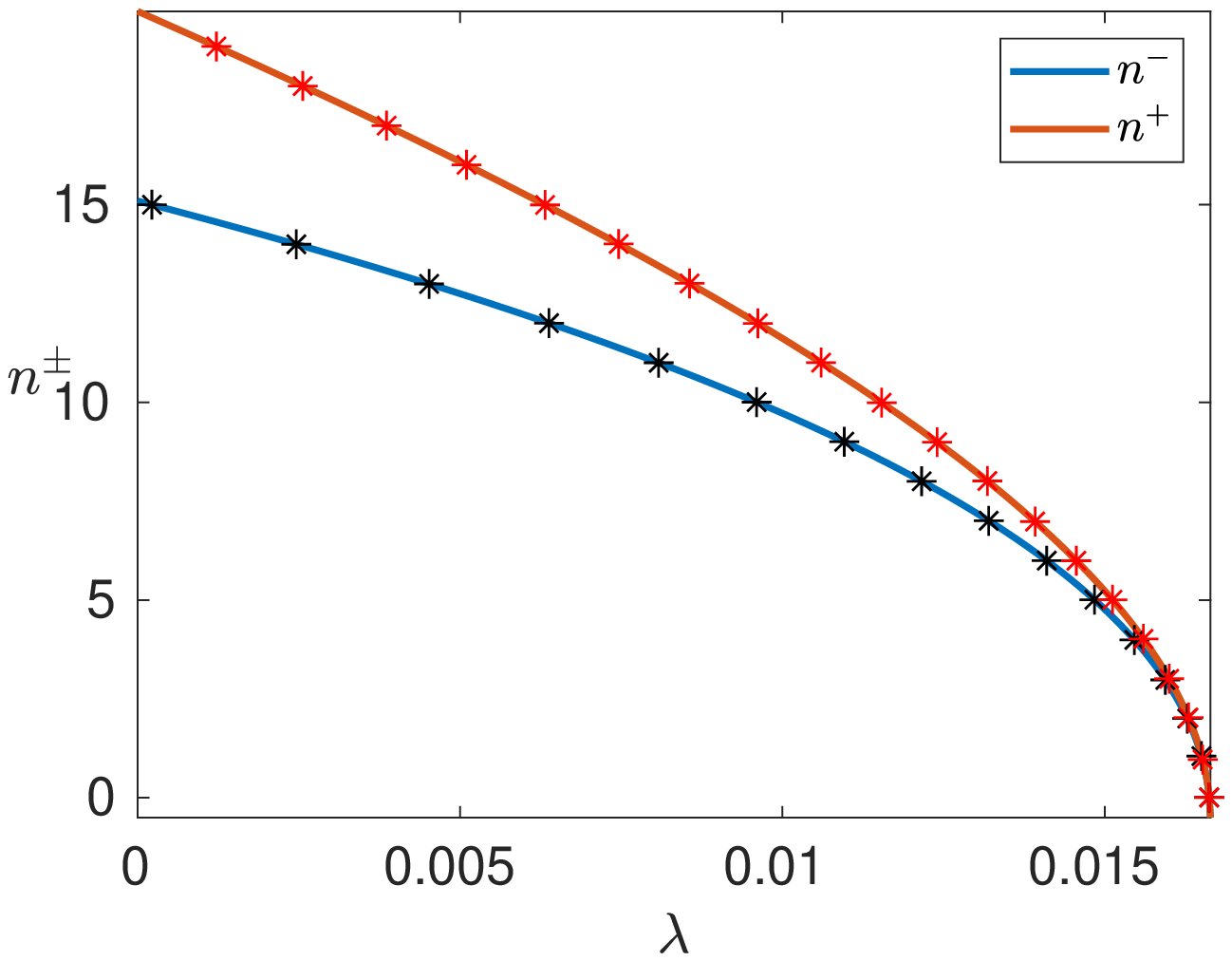}
    \includegraphics[width=0.45\textwidth]{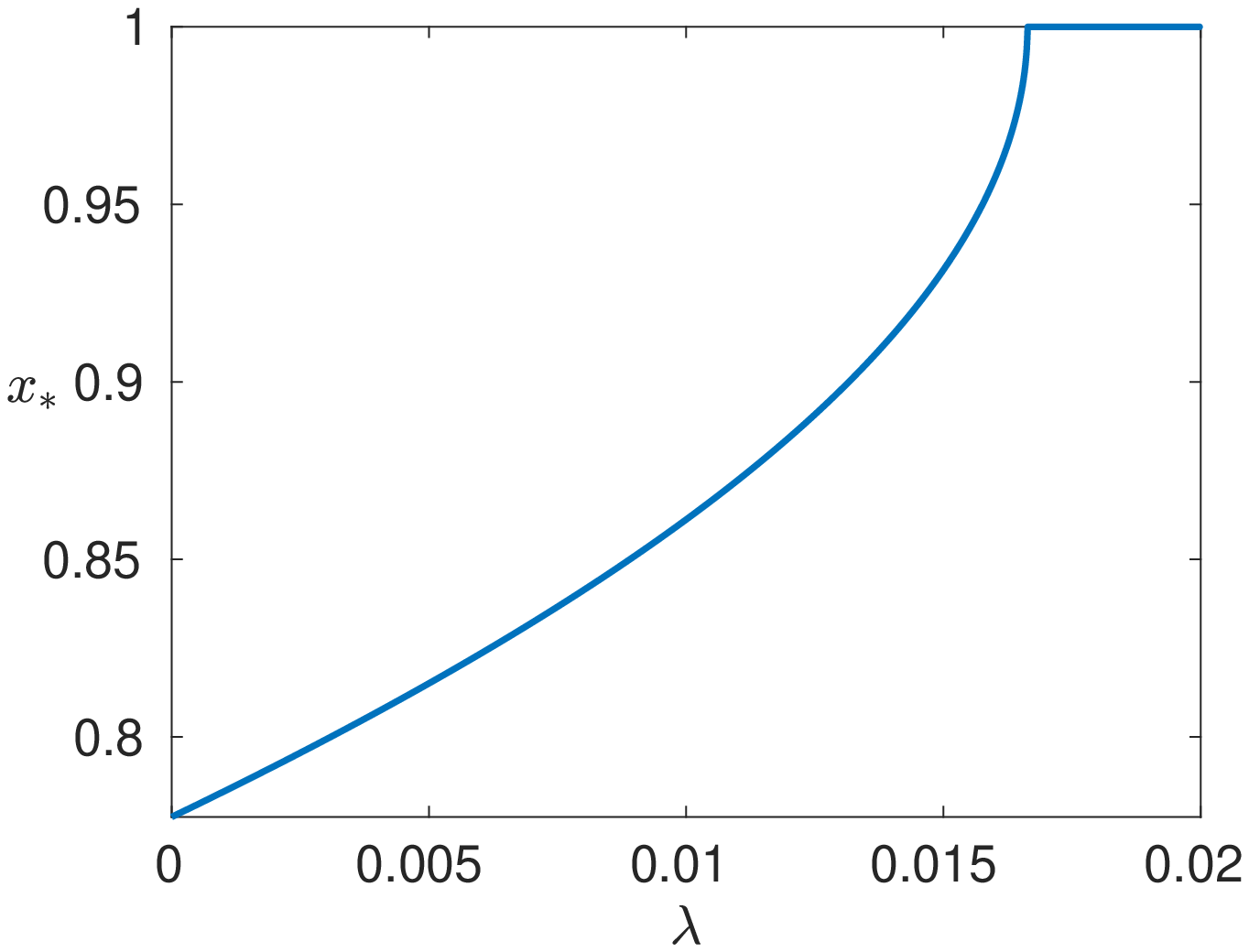}
    
    (a) \hspace{7cm}    (b)
    \caption{Evaluating \eqref{hetselect} reveals the possible discrete modes $n^\pm$ on both branches of WKBJ solutions, which are plotted using asterisks in (a). In (b), we plot the position of singular points $x_*(\lambda)$ demarcating the boundary of $\mathcal{T}_\lambda = [x_*,1]$ as a function of the growth rate $\lambda$, and corresponding to a shrinking $\mathcal{T}_\lambda$ which vanishes when $x_*=1$.}
    \label{Fig.GrowthSupp}
\end{figure}

  \begin{figure}
    \includegraphics[width=0.45\textwidth]{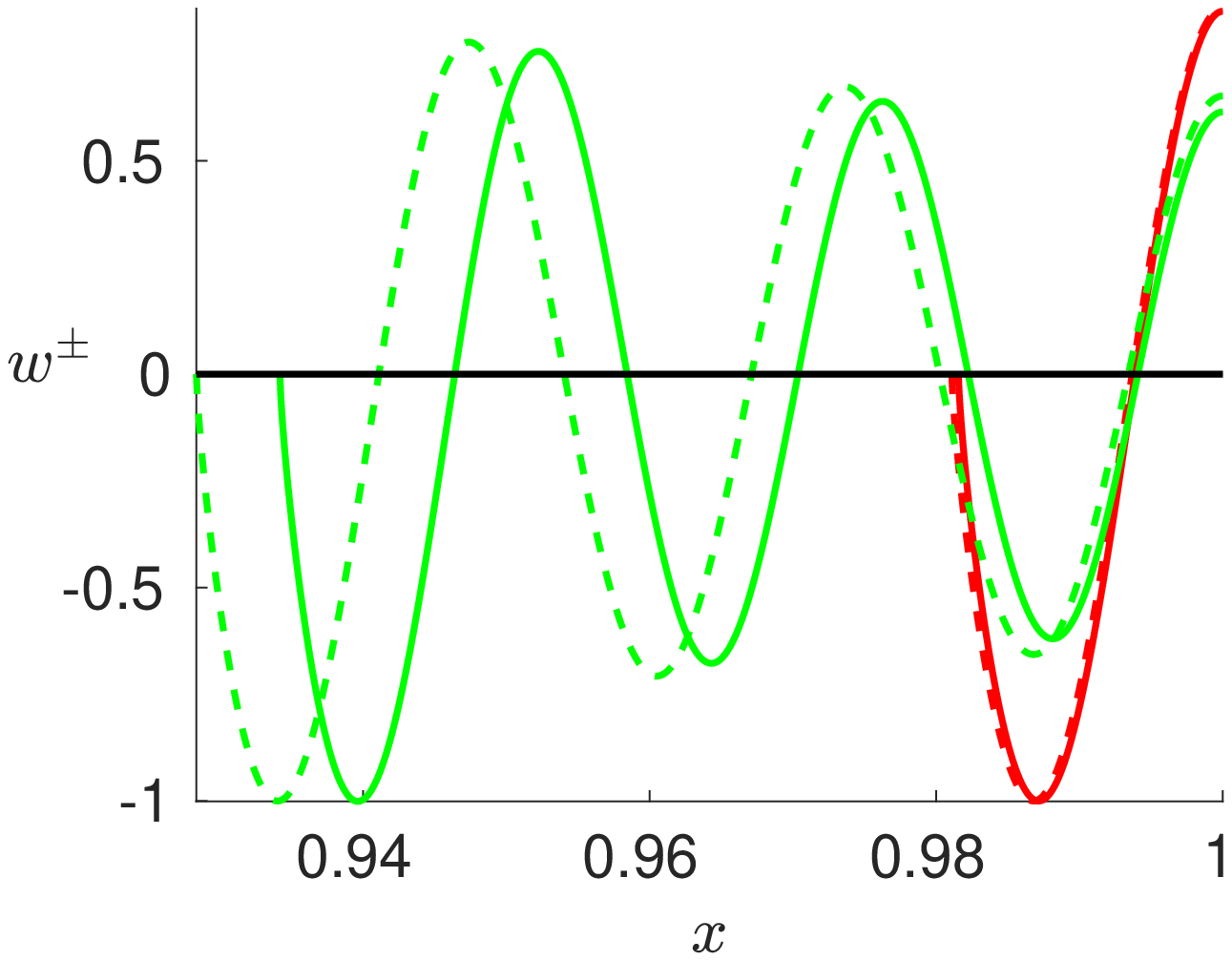}
    \includegraphics[width=0.45\textwidth]{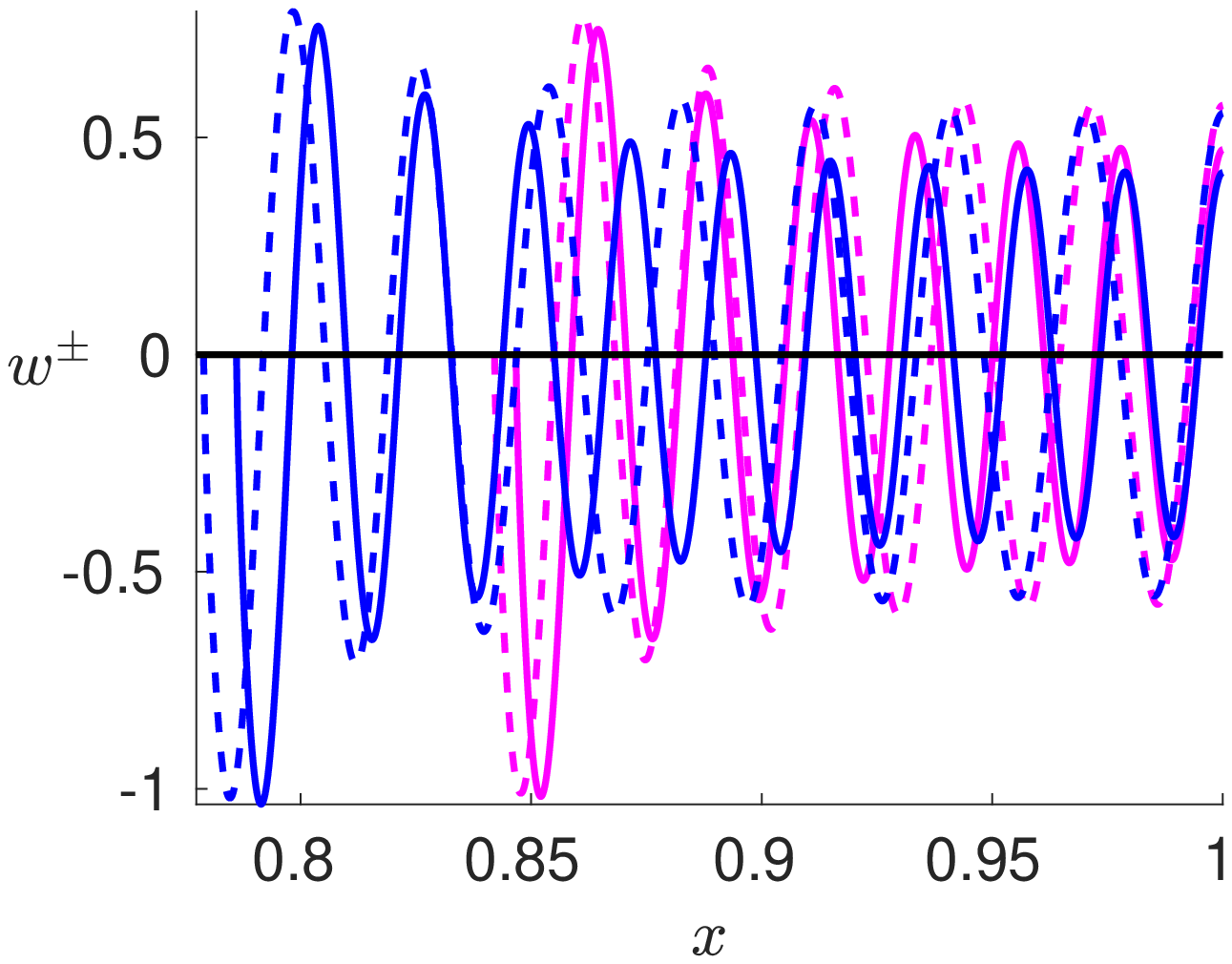}
    
    (a) \hspace{7cm}    (b)
    \caption{ We plot the first component of modes given by \eqref{eq:PiecewiseWKBsolu_right}, associated with $u_1$, corresponding to parameters as in Fig.~\ref{fig1}(a). We plot modes for the positive branch $\mu_\lambda^+$ (solid lines), with $n^+=1,5,13,19$ in red, green, purple, and blue respectively, as well as the modes corresponding to the negative branch $\mu_\lambda^-$ (dashed lines), with $n^-=1,5,11,15$ in red, green, purple, and blue respectively. The smaller two mode numbers are shown in (a), and the larger two in (b). We remark that the smallest and largest values of $n^\pm$ correspond to the maximal and minimal mode numbers along each branch, and the other two mode numbers for each branch are chosen to have similar values of $\lambda$. Note the shrinkage of the support of each mode with increasing $n$, and in particular the difference in the axes for each plot. }
    \label{Fig.Eigenmodes}
\end{figure}

In this particular example there are $19$ unstable WKBJ eigenmodes on the branch corresponding to $\mu_\lambda^+$ and 15 on the branch corresponding to $\mu_\lambda^-$, all of which follow from computing $\lambda$ from \eqref{hetselect} as shown in Fig.~\ref{Fig.GrowthSupp}(a). Note that the right-most value of $n^\pm$ (here shown as a continuous interpolation) corresponds to $\lambda=0$, i.e.~the boundary of $\mathcal{T}_0$ , and that indeed the first unstable WKBJ mode (on the negative branch in this particular example) appears near this boundary as predicted due to the small values of $\epsilon$. We depict four of these modes from each branch in Fig.~\ref{Fig.Eigenmodes}, noting that they each have a support which increases with $n$. We also see in Fig.~\ref{Fig.GrowthSupp}(a) how the growth rate is related to the discrete values of $n$, and how the support of a corresponding mode changes with its growth rate $\lambda$ in Fig.~\ref{Fig.GrowthSupp}(b). Fig.~\ref{Fig.GrowthSupp} directly evidences the predictions from Propositions \ref{c1} and the results in SI Section \ref{App.Lambdan}, as we see the support of distinct modes decrease with increasing $\lambda$. Finally we remark that the vector ${\bf p}^*$ in Equation \eqref{eq:WKBsolu} (computed numerically) is negative in its first component, and positive in its second (as expected for a cross-kinetic system like Schnakenberg), and both components are essentially constant in space, varying by less than $1\%$ of their magnitude across the domain. 

\section{Discussion}\label{Discussion}
We have analyzed two-component heterogeneous reaction-diffusion systems in order to justify the use of \emph{local} Turing conditions which are commonly employed in the literature, and given a deeper insight into how heterogeneity changes the structure of patterned states. Using a WKBJ ansatz, we have shown that \emph{local} conditions are valid, provided that the heterogeneity is slowly varying. { Additionally, we demonstrated that these unstable modes are supported in distinct regions of the domain with different growth rates, and without a well defined wavenumber;  this leads to the commonly-observed amplitude variations reported in the literature.   This is in contrast to the homogeneous case, where the pattern modes have a well defined wavenumber and occupy the whole space (see Figure \ref{RDvsPI}(b))}. We illustrated our analytical predictions using a simple model in Section \ref{Illustration}. Much more complicated heterogeneities and reaction-diffusion systems, such as those explored in \cite{page2005complex} were also used to verify the analytical predictions in more complicated cases, such as when $\mathcal{T}_0$ is no longer a single interval. Nevertheless, the instability criteria work well for suitably small $\epsilon$ such that the heterogeneity does not vary faster than $O(1/\epsilon$). While we can enumerate the unstable modes and compute their growth rates, we remark that there is no obvious generalization of wavelength or frequency in this setting; unstable modes, and fully developed patterns, tend to exhibit large varying oscillations throughout a heterogeneous region of space.

Alongside generalizing the classical Turing conditions to the case of spatially heterogeneous systems, our analysis suggests several further questions to pursue. We have assumed that the local steady state is stable in the absence of diffusion throughout the domain, but it may be possible that diffusion could in fact \emph{stabilize} the solution of a heterogeneous reaction-diffusion system which is locally unstable (in the absence of diffusion) in only part of its domain, leading to a patterned state. Additionally, it is known that rapidly varying heterogeneities can substantially impact the ability of a reaction-diffusion system to admit patterns, and the qualitative features that such patterns exhibit \cite{rudiger2007theory,stephetero}, and this remains to be explored within the present framework. Finally, while the results in SI Section \ref{App.Lambdan} allow us to conjecture about the envelope of solutions via the growth rate of distinct unstable modes, these remarks have not been rigorously justified. Demonstrating properties of these envelopes mathematically would require extending the framework of weakly nonlinear analysis \citep{crawford1991introduction, cross1993pattern, wollkind1994weakly, stephenson1995weakly, chen2019non} to the heterogeneous setting, and is beyond the scope of this paper. 

In addition to these mathematical extensions, one could apply these results directly to biological patterning situations, such as successive patterning due to reaction-diffusion mechanisms on different timescales, or to the combination of theories of positional information and reaction-diffusion (see Fig.~\ref{RDvsPI}). Originally, the well-known Gierer-Meinhardt model developed in \cite{gierer1972theory} contained a spatial heterogeneity representing a precursor pattern from a previous pattern forming event. Such a situation could be directly captured by considering distinct reaction-diffusion processes occurring at different time points in development, or on different temporal and spatial scales. Alternatively, one can posit a positional information framework as the origination of spatial structure, such as in delineating different patterning fields from one another, and let reaction-diffusion theory produce additional periodic patterning within this heterogeneous domain, as suggested in \cite{green2015positional}. hence this paper presents a first step toward theoretically understanding the evolution of one pattern into another, but much more work must be done linking to experimental studies to justify such a theory of morphogenesis.

While the WKBJ-based approach we have employed is potentially extendable to multi-species or multi-dimensional systems, the calculations become increasingly complicated. Real chemical and biological systems are composed of many different chemical species, and few developmental phenomena are faithfully captured by a single spatial dimension. Additionally, we remark that our analysis presented in SI Section \ref{Sec.PropsOfWKB} only shows that continuous modes can be defined across singularities at leading order, though fully resolving the boundary-layer structure across these singularities is beyond our present scope. Nevertheless, the results we have presented here will remain valid even with such refinements. We also anticipate that these results are indicative of Turing instabilities in heterogeneous systems in higher dimensions, or with three or more species. Specifically, spatial regions which satisfy local Turing conditions should admit patterned solutions (distinct from the ambient heterogeneity) if these regions are sufficiently large, and the spatial heterogeneity is sufficiently smooth. Preliminary numerical investigations in two and three dimensions suggest this is true, and a valuable extension given the biological motivations for the theory. The framework presented here is a first step in understanding how one patterned state arises from another, and in elucidating the more nuanced roles that reaction and diffusion play in development and analogous systems with heterogeneous instabilities. As Turing said, though under different circumstances \cite{machinery1950computing}, ``We can only see a short distance ahead, but we can see plenty there that needs to be done."

\section{Acknowledgements}
A.L.K. and E.A.G. are grateful for support from BBSRC grant BB/N006097/1;  V.K. is grateful for support from European Regional Development Fund-Project ``Center for Advanced Applied Science" (No. CZ.02.1.01/0.0/0.0/16\_019/0000778) and the Mathematical Institute at the University of Oxford. In compliance with BBSRC's open access initiative, the data in this paper is available from http://dx.doi.org/xx.xxxx/xxxxxxxxxxxxxxxxxx.

\bibliographystyle{apsrev}
\bibliography{refs}

\newpage

%%%%%%%%%%%%%%%% TO APPENDIX \ref{App.RederivationOfClassDDI}.
\section*{Supplementary Information for the \emph{Journal of the Royal Society Interface} article ``From One Pattern into Another: Analysis of Turing Patterns in Heterogeneous Domains via WKBJ" by A. L. Krause, V. Klika, T. E. Woolley, and E. A. Gaffney}
\beginsupplement

\section{Rederivation of Turing instability with spatial homogeneity.}  \label{App.RederivationOfClassDDI}
The standard derivation of the Turing conditions in the homogeneous setting with Neumann boundary conditions considers the separable Fourier solution 
\beq \label{Fs} {\bf w} \propto \exp(\lambda t)\cos(k x),
\eeq
with wave number $k$, and finds the associated growth rate, $\lambda$. The conditions for a Turing instability then arise from the requirement that: 
\begin{itemize}
    \item [(i)]
there is stability when $k=0$, indicating a stable steady state without diffusion, 
\item[(ii)] a range of $k\neq 0$ generates an instability, at least providing  $k/\pi$ is a non-zero integer within this range. 
\end{itemize}
In the heterogeneous case, the conditions associated with stability in the absence of diffusion are derived analogously to the homogeneous case. However, for instability,  the Fourier solutions do not decouple and we seek an alternative approach. Proceeding, we firstly summarise the calculation of the homogeneous Turing conditions, where the fundamental equation arising from the substitution of (\ref{Fs}) into (\ref{weqn1})
is given by 
%so we  instead fix the growth rate and then search for admissible unstable solutions,} restricting the choice of $\lambda$ \emph{a posteriori} in an analogous way to wavemode selection on a finite domain. Hence,   we briefly consider this approach for  the homogeneous Turing conditions,  for generalisation to consider spatial heterogeneity.} 
%This approach is different from the classical picture, where we are looking for conditions that guarantee a value of $\lambda$ with positive real part. Instead, we think of $\lambda$ as a parameter, and ask what conditions are consistent with it having a positive real part, and being permissible in the sense of admitting a real unstable mode.
%{\bl In particular, the} fundamental equation in the homogeneous case is the solvability condition for eigenmode $k$
\beq \label{fe1}   \mbox{det} [\epsilon^2k^2 {\bf D} -{\bf J} + \lambda {\bf I} ] = \mbox{det} [\epsilon^2k^2 {\bf D} -{\bf J}_\lambda] = 0, 
\eeq
where we denote ${\bf J_\lambda} = {\bf J} - \lambda{\bf I}$. This condition is equivalent to 
$$ \mbox{det} [ -{\bf D}^{-1}{\bf J_\lambda} + \epsilon^2k^2 {\bf I}   ] = 0.$$
Hence, given $\epsilon$ and the wave number $k$, we can determine the growth rate $\lambda$. For a perturbation to grow we require values of $k^2$ and $\lambda$ such that 
\beq 
\label{fcc} 
\mbox{Re}(\lambda)>0, ~~~~ \mbox{for} ~~~k^2=n^2\pi^2>0, 
\eeq 
with $n$ a non-zero integer, subject to Equation \eqref{fe1}. Instead of following the normal approach where we vary $k$ to ensure $\Re(\lambda)>0$, we can instead vary $\lambda$ to deduce conditions under which requiring $k^2$ to be real and positive implies the normal Turing conditions. The relationship between $\lambda$ and $k^2$ is computed from the dispersion relation \eqref{fe1}.

{\bf Non-real growth rates.}
Permissible  values of $k^2$ are real and positive, or else we could not satisfy $k^2=n^2\pi^2$ for an integer $n\neq 0$. Thus, we can exclude cases where  $\epsilon^2k^2$ is not strictly real. We also neglect cases where  Re$(\lambda)<0$  as we are only interested in instability. However we have
\beq \label{trhom}
-\mbox{tr}(\epsilon^2k^2 {\bf D}-{\bf J})= [f_u+g_v -\epsilon^2k^2(1+d)] < 0,\eeq
for permissible $k^2$, given that the homogeneous steady state is stable, so that 
tr$({\bf J})= f_u+g_v <0$. We also have from \eqref{fe1} that $\lambda$ satisfies,
\beq \label{lamhom}
\lambda = -\mbox{tr}(\epsilon^2k^2 {\bf D}-{\bf J})
\pm\sqrt{[\mbox{tr}(\epsilon^2k^2 {\bf D}-{\bf J})]^2-4 \mbox{det}[\epsilon^2k^2 {\bf D}-{\bf J}] }.
\eeq

Hence, if $\lambda$ is non-real, its real part is negative as follows from \eqref{trhom}-\eqref{lamhom}. Thus, without loss of generality, we can consider real $\lambda$, as complex growth solutions with permitted wave numbers, if they exist, are stable.

{\bf Real growth rates.} Recalling the notation ${\bf B}_0 = {\bf D}^{-1}{\bf J}$, the transition to instability occurs when $\lambda=0$, whence
\beq \label{homogdisp}  \epsilon^2k^2 &=& \frac 12 \left[\mbox{tr}( {\bf B}_0)
\pm\sqrt{[\mbox{tr}( {\bf B}_0)]^2-4 \mbox{det}({\bf B}_0) }\right],
\eeq
where the two roots $\epsilon^2k^2$ are the eigenvalues of ${\bf B}_0$. Therefore, to generate an inhomogeneous instability consistent with the stability of the zero mode (conditions \eqref{stab0}) we require
\beq \label{fc}  \mbox{tr}{( {\bf B}_0)}
-\sqrt{[\mbox{tr}( {\bf B}_0)]^2-4 \mbox{det}({\bf B}_0) }
>0.
\eeq 

There are two roots for $\epsilon^2k^2$ and {\it these are the eigenvalues of} 
$ {\bf B}_0$. If both are negative, there is no permissible value of $k^2$, and thus there is no instability. If the smaller eigenvalue is negative and the larger is positive, then from the shape and behaviour of Re$(\lambda)$ as  a function of $\epsilon^2k^2$, one must have Re$(\lambda) = \lambda>0$
for $\epsilon^2k^2=0$. However, this possibility will be excluded as we require the mode associated with $k=0$ to be stable. Thus we require that inequality \eqref{fc} holds in order to generate an inhomogeneous instability consistent with the stability of the zero mode. In analogy to the standard derivation of the Turing condition, the requirement that $k/\pi$ is a non-zero integer is only considered \emph{a posteriori}. 
From the conditions for the stability of the zero mode, \eqref{stab0},  we have that 
$\mbox{det}({\bf J})>0$, and hence  $\mbox{det}({\bf D}^{-1}{\bf J})>0$. The remaining conditions for an inhomogeneous instability are then
\beq \label{tc230} \mbox{tr}( {\bf D}^{-1}{\bf J}) > 0, ~~~~~~ [\mbox{tr}( {\bf D}^{-1}{\bf J})]^2-4 \mbox{det}({\bf D}^{-1}{\bf J}) > 0 ,
\eeq 
which are equivalent to the standard Turing conditions \eqref{instabk} given in Criterion \eqref{homogturing}. 

%%%%%%%%%%%%%%%%

\section{Singularities of WKBJ modes} \label{Sec.PropsOfWKB}

Here we show properties of the solution near internal singular points in detail, denoting such a point as $x_*$. At any singular point, where ${\bf s}_*^T {\bf p}_*=0$, the expression for $Q_0$, \eqref{Q0}, becomes ill-defined, and hence we examine the structure of the solution near such a singular point.

With ${\bf J}_\lambda^*$ denoting $ {\bf J}_\lambda (x_*)$ and similarly ${\bf B_\lambda^*}={\bf B}_\lambda(x_*)$, we focus on the case when ${\rm tr}({\bf J}^*)<0$,  ${\rm det}({\bf J}^*)>0$, ${\rm tr}({\bf B}^*_\lambda)>0$ which is the case of interest as this will be the boundary of $\mathcal{T}_\lambda$ (see Proposition \ref{p4}). Further, the zero of $ [\rm{tr}( {\bf B}_\lambda)]^2-4 {\rm det}({\bf B}_\lambda) $ is generically a simple one at $x=x_*$, as a non-simple zero would require mathematical fine-tuning in the model and parameter choices for smooth kinetic functions.
 Then for fixed $y\neq x_*$ with ${\bf s}^T_* {\bf p}_*\neq 0$ in $(x_*,y)$, the integral $$ \exp \left[ \int_x^{y} \frac{{\bf s_*}(\bar{x})^T{\bf p}_*'(\bar{x})}{{\bf s_*}(\bar{x})^T{\bf p}_*(\bar{x})} \mathrm{d}\bar{x} \right]$$
 has a singularity which scales with $1/|x-x_*|^{1/4}$ as 
$x\rightarrow x_*$. 

\begin{proposition} \label{p9} 
Let $\lambda$ be a non-negative real growth rate. We assume
 ${\rm tr}({\bf J}^*)<0$,  ${\rm det}({\bf J}^*)>0$, and   ${\rm tr}({\bf B}_\lambda^*)>0$ with  ${\bf J}_\lambda^*$ denoting $ {\bf J}_\lambda 
(x_*)$. Additionally, we assume that the zero of $ [\rm{tr}( {\bf B}_\lambda)]^2-4 {\rm det}({\bf B}_\lambda)) $ is a simple one at $x=x_*$. Then with fixed $y\neq x_*$ the integral
$$ \exp \left[ \int_x^{y} \frac{{\bf s_*}(\bar{x})\cdot{\bf p}_*'(\bar{x})}{{\bf s_*}(\bar{x})\cdot{\bf p}_*(\bar{x})} \mathrm{d}\bar{x} \right]$$
 has a singularity which scales with $1/|x-x_*|^{1/4}$ as 
$x\rightarrow x_*$. 

\end{proposition}

\begin{proof}
By \eqref{ie3} we have 
$$ \mu^\pm_\lambda(x_*) = \frac 12 {\rm tr}({\bf B}_\lambda^*), $$
which is a double root at $x=x_*$, so we have,
 \begin{equation}\label{Bmat}
[-\mu^\pm_\lambda(x_*){\bf I}  +{\bf D}^{-1}{\bf J}_\lambda]=
\left( \begin{array}{cc} f_u -\lambda- \mu^\pm_\lambda  & f_v \\ \dfrac{g_u}d & \dfrac{g_v -\lambda}d  -\mu^\pm_\lambda  \end{array} \right), 
~~~~~( f_u -\lambda -\mu^\pm_\lambda)\left(\dfrac{g_v -\lambda}d -\mu^\pm_\lambda \right)-\dfrac 1 d f_vg_u=0.
\end{equation}

With sign choices that are without loss of generality,  we compute
$$ {\bf p}_* (x)= \frac{1}{R_p} \left( \begin{array}{c} -f_v \\ f_u -\lambda- \mu^\pm_\lambda   \end{array} \right), ~~~~~~~~~
R_p= \left(f_v^2+|f_u -\lambda- \mu^\pm_\lambda|^2\right)^{1/2},
$$ 
with 
$$   {\bf s}_* (x)= \frac{1}{R_s} \left(   -g_u/d,  f_u -\lambda- \mu^\pm_\lambda    \right), ~~~~~~~~~
R_s= \left(\left( \dfrac{g_u}d\right)^2+|f_u -\lambda- \mu^\pm_\lambda|^2\right)^{1/2}.
$$
We will need $R_p \sim O(1)$ near $x=x_*$, and so must show that $R_p(x_*) \neq 0$. We proceed by contradiction and assume that $R_p=0$ at $x=x_*$. Then $f_u -\lambda- \mu^\pm_\lambda=0=f_v$. So therefore $f_u > \lambda > 0$, but we have det$({\bf J}^*)  = f_ug_v > 0$ and tr$({\bf J}^*)  = f_u+g_v < 0$, which cannot be simultaneously satisfied, and hence we must have $R_p\neq0$ at $x=x_*$. An analogous proof also shows that $R_s\neq 0$ at $x=x_*$.

%Proceeding, we have 
%$${\bf s}_* (x)\cdot {\bf p}_* (x) = \frac1{R_pR_s}\left[
%\dfrac 1 d f_vg_u +(f_u -\lambda- \mu_\lambda)^2
%\right].$$
Letting $x=x_*+X$ with $|X| \ll 1$,  and defining 
$$\alpha_\lambda = \frac{ \partial }{\partial x}
\left([\mbox{tr}( {\bf B}_\lambda)]^2-
4 \mbox{det}({\bf B}_\lambda)
\right)
\bigg|_{x=x_*},$$ 
we have from equation~\eqref{ie3}  that near $x=x_*$
\beq  \mu^\pm_\lambda(x) := \frac 12\left\{  \begin{array}{ll}  \left[\mbox{tr}( {\bf B}^*_\lambda)
\pm |\alpha_\lambda  X| ^{1/2} +O(X) \right]  &   \alpha_\lambda  X >0\\
 \left[\mbox{tr}( {\bf B}^*_\lambda)
\pm i |\alpha_\lambda X|^{1/2} +O(X) \right] & \alpha_\lambda  X < 0
\end{array}
\right\} := \mu_{\lambda}^{0*} +  \frac 12\left\{  \begin{array}{ll}   
\pm |\alpha_\lambda  X| ^{1/2} +O(X)    &   \alpha_\lambda  X >0\\ 
\pm i |\alpha_\lambda X|^{1/2} +O(X)   & \alpha_\lambda  X < 0
\end{array}
\right\},
\eeq 
where $\mu_{\lambda}^{0*} = \mbox{tr}({\bf B}^*_\lambda)$ is constant in $X$. We do not consider the degenerate case of  $\alpha_\lambda=0$ as this corresponds to a non-simple root of det$([-\mu^\pm_\lambda(x){\bf I}  +{\bf D}^{-1}{\bf J}_\lambda])$ at  $x=x_*$. From this expansion, 
we have (denoting derivatives with respect to $X$ as $' \equiv \partial_X)$
\beq ( \mu'^\pm_\lambda)(x)  =   \frac 1 4 \left\{  \begin{array}{ll}   
\pm \left | \dfrac{\alpha_\lambda}  X \right| ^{1/2} +O(1)    &   \alpha_\lambda  X >0\\ 
\pm i \left |\dfrac{\alpha_\lambda}  X \right|^{1/2} +O(1)   & \alpha_\lambda  X < 0
\end{array}
\right\} . 
\eeq 
Specialising in the first instance to the case $ \alpha_\lambda  X >0$, and on noting 
$f_u=f_u^*+O(|X|),$ $f_v=f_v^*+O(|X|)$  near $x=x_*$, while $R^*_p=R_p(x_*)\sim O(1), ~R^*_s=R_s(x_*)\sim O(1)$, the contraction of ${\bf s}_*(x)$  and ${\bf p}_*(x)$ yields 
\begin{eqnarray*}  {\bf s}_* (x)\cdot {\bf p}_* (x) &=&  \frac1{R_pR_s}\left[
\dfrac 1 d f_vg_u +(f_u -\lambda- \mu_\lambda)^2
\right]   \\ &=&  {\bf s}_* (x_*)\cdot {\bf p}_* (x_*) 
\mp \frac1{R^*_pR^*_s(1+O(|X|^{1/2}))}\left[ (f^*_u-\lambda -\mu_\lambda^{0*})|\alpha_\lambda X|^{1/2} +O(|X|) 
\right]
\\  &=&   
\mp \frac1{R^*_pR^*_s}  (f^*_u-\lambda -\mu_\lambda^{0*})|\alpha_\lambda X|^{1/2} +O(|X|), 
~~~~ \mbox{as~} X \rightarrow  0.
\end{eqnarray*}

Further, on  differentiating ${\bf p}_*(x)$, one finds 
$$ 
{\bf p}_*' (x)= \frac{1}{R_p} \left( \begin{array}{c} -f_v \\ f_u -\lambda- \mu_\lambda   \end{array} \right)' -\frac{R'_p}{R_p^2}\left( \begin{array}{c} -f_v \\ f_u -\lambda- \mu_\lambda   \end{array} \right)
= \frac  1 {4R_p^*}  \left( \begin{array}{c} O(1) \\\mp 
 \left | \dfrac{\alpha_\lambda}  X \right| ^{1/2}
\end{array} \right) +O(1) - \frac{R'_p}{R_p}{\bf p}_* (x).$$
 Contracting with ${\bf s}_*(x)$ yields 
$${\bf s}_*(x) \cdot 
{\bf p}'_* (x) = \mp\frac{1}{4R_p^*R_s^*}\left[ f^*_u -\lambda- \mu_\lambda^{0*}\right] \left| \dfrac{\alpha_\lambda}  X \right| ^{1/2} +O(1) -  \frac{R'_p}{R_p}{\bf s}_*(x)\cdot{\bf p}_* (x) =  \mp\frac{1}{4R_p^*R_s^*}\left[ f^*_u -\lambda- \mu_\lambda^{0*}\right] \left| \dfrac{\alpha_\lambda}  X \right| ^{1/2} +O(1), 
$$
on noting that 
$$ R_p' \sim O\left(\left| \dfrac{\alpha_\lambda}  X \right| ^{1/2}\right), ~~~ {\bf s}_* (x)\cdot {\bf p}_* (x)     \sim O\left(|X|^{1/2}\right), ~~~~ \mbox{as~} X \rightarrow 0.$$ Hence 
\begin{eqnarray} \label{finratio}
\dfrac{{\bf s}_*(x) \cdot 
{\bf p}'_* (x)}{ {\bf s}_* (x)\cdot {\bf p}_* (x)} = \frac{1}{4|X|}\left(1+O(|X^{1/2}|)\right), 
\end{eqnarray}
and the above calculations hold for arbitrary $\alpha_\lambda \neq0$ and thus equation~(\ref{finratio})
also holds for $\alpha_\lambda X <0.$ In turn, we have for $ x=x_*+X$,  $X>0$, 
\beq  \nonumber 
\frac{Q_0(x_*+X)}{Q_0(1)} \propto \exp \left[ -\int^{x~_*+X} _1 \frac{{\bf s_*}(\bar{x})\cdot{\bf p}_*'(\bar{x})}{{\bf s_*}(\bar{x})\cdot{\bf p}_*(\bar{x})} \mathrm{d}\bar{x} \right]  = \exp \left[\int_{x~_*+X} ^1 \frac 1 {4(\bar{x}-x_*) }     +O\left( \frac 1 {(\bar{x}-x_*)^{1/2}} \right)
  \mathrm{d}\bar{x} \right] \sim O\left(\dfrac 1 {|X|^{1/4}}\right),
\eeq  
and the same scaling holds for $X<0.$ 
\end{proof}

On approaching such a singular point $x_*$, solutions and their derivatives become unbounded and the asymptotic assumptions inherent in the WKBJ approximation (that the diffusion term is subleading), breaks down. In contrast, the second derivative of such solutions near $x\approx x_*$ will scale with $1/|x-x_*|^{9/4}$, and hence the transport term is no longer asymptotically small when $|x-x_*|\sim \epsilon^{8/9}$. Hence boundary layers are present around $x_*$. However, the interior boundary layer problem is not tractable analytically, and thus we only consider the outer WKBJ solutions.  Nonetheless, we require boundedness of the outer solutions on approaching the boundary layer, otherwise such solutions will be arbitrarily large for sufficiently small $\epsilon$.   In turn the outer region is valid for $|x-x_*|\sim \epsilon^{4/9} \gg  \epsilon^{8/9}$, where the WKBJ solution scales with $1/|x-x_*|^{1/4} \sim 1/\epsilon^{1/9} \rightarrow \infty$ as $\epsilon \rightarrow 0$. Solution boundedness requires the expression in \eqref{eq:WKBsolu} to take the form of a $\sin$ function near the singular point $x_*$, and a $\cos$ function is used at a zero-flux boundary.

%Implicit in the above working, the WKBJ solutions were assumed to hold across the domain. While bounding $\mu_\lambda$ away from zero has been feasible, it is possible for  ${\bf s_*}(x)\cdot {\bf p}_*(x) =0,$ for values of $x_*\in[0,1]$  where ${\bf s_*}(x),~ {\bf p}_*(x)$ are respectively the left and right eigenvectors of zero eigenvalue for  $[-\varphi'^{2} {\bf I}  +{\bf D}^{-1}{\bf J}_\lambda]=[-\mu^\pm_\lambda(x){\bf I}  +{\bf D}^{-1}{\bf J}_\lambda]$. We now consider the impact of a zero of ${\bf s}_*(x_*)\cdot {\bf p_*}(x_*)$ with 
%$x_*\in(0,1).$ 

\section{Relationship between $n^\pm$ and $\lambda$ and the support of non-trivial WKBJ modes}\label{App.Lambdan}

Here we show that $\lambda$ decreases with $n^+$ in the positive branch of WKBJ solutions, and outline how the negative branch behaves.
 \begin{proposition} \label{p11} 
The value of the non-negative growth rate $\lambda$ decreases with $n^+$ for the positive branch of WKBJ solutions. 
\end{proposition}

\begin{proof}
We proceed by differentiating the fundamental constraint \eqref{eq:FundConstr} with respect to $n^\pm$ to find
\begin{equation}\label{diffn}
\lambda'(n^\pm) \left[b'(\lambda)\sqrt{\mu_\lambda^\pm(b(\lambda))}- a'(\lambda)\sqrt{\mu_\lambda^\pm(a(\lambda))}+\int_{a(\lambda)}^{b(\lambda)}\frac{\partial_\lambda(\mu_\lambda^\pm)(\bar{x})}{2 \sqrt{\mu_\lambda^\pm(\bar{x})}} \mathrm{d}\bar{x}\right] = \pi\epsilon > 0.
\end{equation}
By Proposition \ref{c1} we have that $a'(\lambda) \geq 0$ and $b'(\lambda) \leq 0$, which implies that the first two terms of \eqref{diffn} are together negative and we must only check the sign of the third term. For the positive branch, this term is negative by the proof of Proposition \ref{newp}, as $\partial_\lambda\mu_\lambda^+<0$, hence, for this branch we must have $\lambda'(n^+) < 0$.
\end{proof}

Any non-trivial WKBJ solution has a support (in space) demarcated by singular points $x_*(\lambda(n^\pm))=:x_*(\lambda)$ or the domain boundaries. %; for a fixed $\lambda$, these singular points are the same for both branches of solutions as $\mu_\lambda^+=\mu_\lambda^-$ due to the multiplicity of the eigenvalues at these boundaries.
Therefore the support of the {$n^+$}-th mode also decreases (or remains the same) as {$n^+$} is increased, due to the monotonicity of $\mathcal{T}_\lambda$. We can conclude that {$\mathcal{T}_{\lambda(n+)}$} shrinks with increasing {$n^+$}, and that the largest permissible $n^+$ will correspond to the smallest value of $\lambda$ and the largest spatial support, whereas the smallest $n^+$ will have the smallest support, but largest growth rate. 

For the negative branch $n^-$, the calculation in the proof of Proposition \ref{p11} reveals that a competition between two terms of different signs takes place, and the overall picture is more complicated. First, note that if there is no singular point within the domain $[0,1]$ for a range of $\lambda$, then the first two boundary terms of \eqref{diffn} vanish and the last term was shown to be positive via the proof in Proposition \ref{newp} for $\mu_\lambda^-$. Hence, in this scenario, $\lambda$ would increase with increasing $n^-$. In the case when there is an internal singular point, we know that $\mathcal{T}_\lambda$ decreases with increasing $\lambda$. Additionally, near the maximal permissible $\lambda$, the support of the fastest growing mode is very small and becomes a strict subset of $[0,1]$. Further, for such maximal admissible $\lambda$, both $n^\pm=1$ as $\mu_\lambda^\pm(x)$ is a continuous function in $\lambda$ and in $x \in [a,b]$ including the boundaries, while we know that $\mathcal{T}_\lambda$ is monotonic in $\lambda$. Hence, the left hand side of equation \eqref{hetselect} is arbitrarily small and as a result $n^\pm=1$ for the largest admissible $\lambda$. Therefore, close to this maximal value, $\lambda$ has to decrease with $n^-$. Finally, as all the terms determining the sign of $\lambda'(n^-)$ have a fixed sign, we know that there is at most one extremum of $\lambda(n^-)$, which then completes the picture for the negative branch.

\begin{figure}
    \centering
    \includegraphics[width=0.48\textwidth]{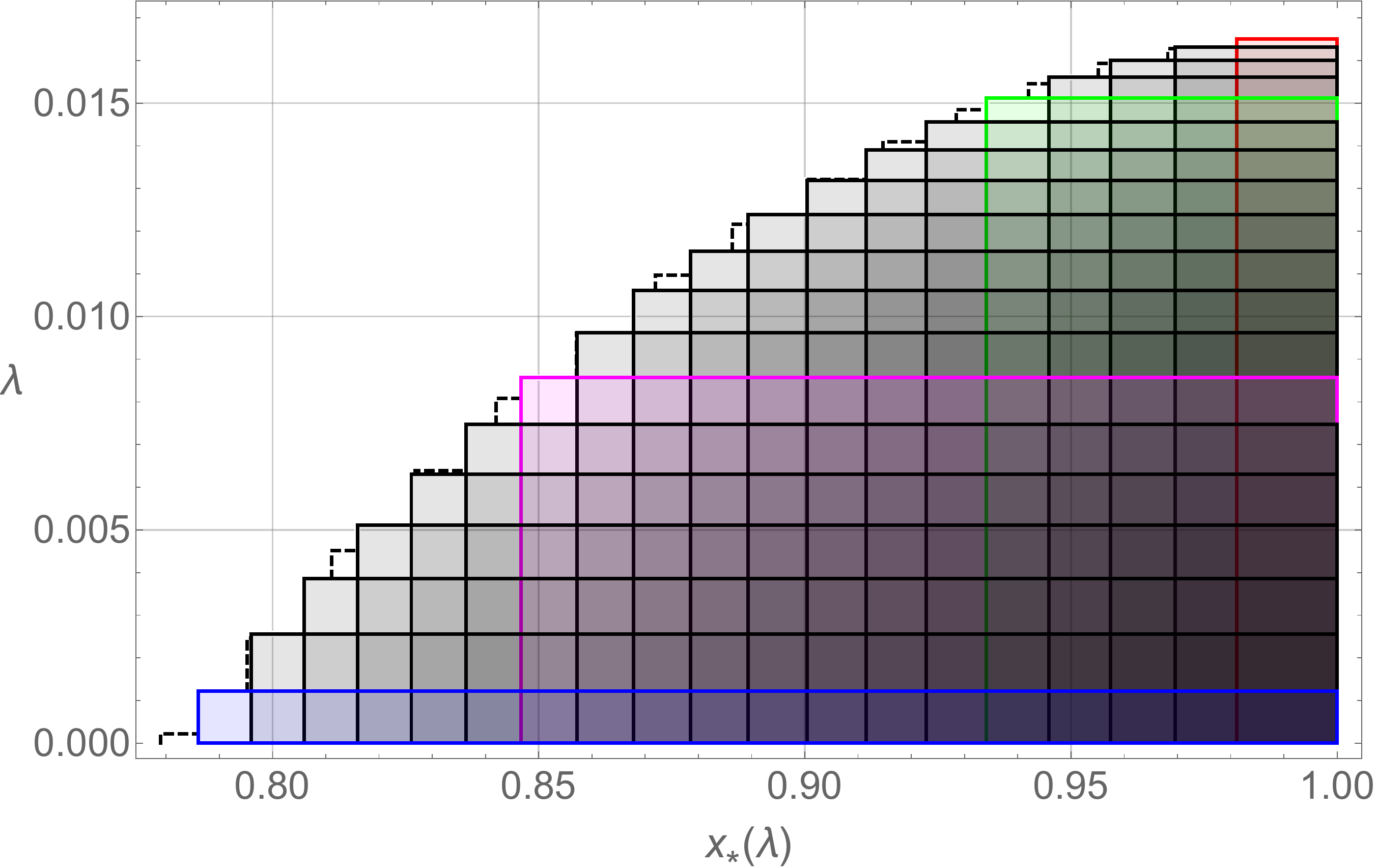}
    \includegraphics[width=0.48\textwidth]{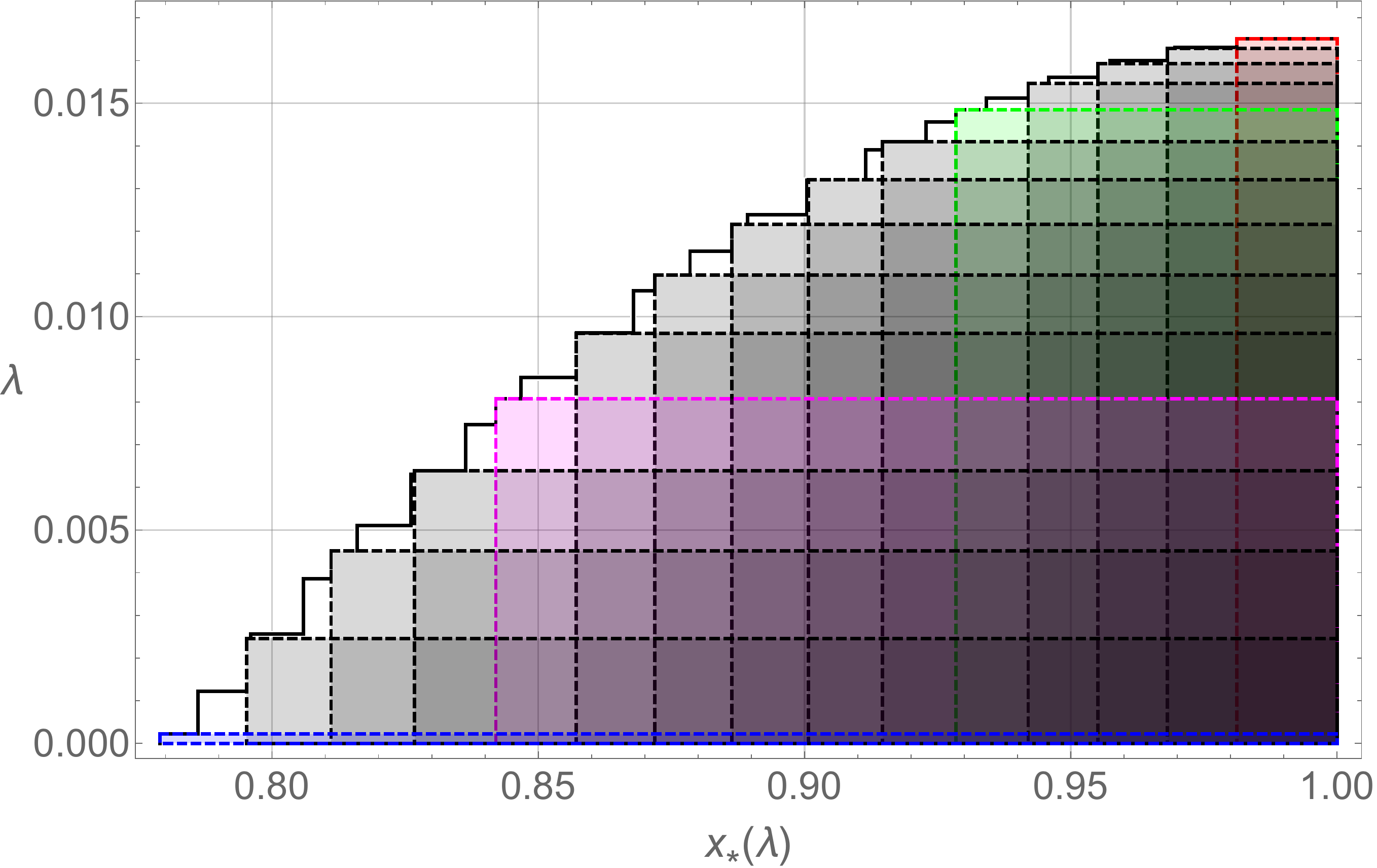}
    
    (a) \hspace{7cm}    (b)
    \caption{The ``structure'' of a patterned state linear analysis showing the intervals where a given WKBJ solution from the positive branch ((a), solid lines) and from the negative branch ((b), dashed lines) will dominate. The discreteness of the steps is highlighted together with the value of growth rate $\lambda$.
      For each mode we plot an opaque rectangle with the horizontal side being the support of the mode,i.e.~the interval $(x_*(\lambda(n^\pm)),1)$, while the vertical side is the value of the growth rate for a given mode $\lambda(n^\pm)$. Note that the highlighted rectangles in colors correspond to the modes depicted in Fig \ref{Fig.Eigenmodes}. Hence, there are subintervals where many modes exist. The envelope of the largest $\lambda$ values then forms then the topmost black line at the boundary, which we conjecture to have a relation to the amplitude envelope of emerging patterned solutions.
    }
    \label{Fig.PrePattern}
\end{figure}

In Fig.~\ref{Fig.PrePattern} we show the support of each discrete mode alongside its corresponding growth rate for both solution branches. We highlight regions corresponding to the four modes of each branch given in Fig.~\ref{Fig.Eigenmodes}. We note in particular that the fastest growing mode in any given spatial region is the mode which is highest in any given region in Fig.~\ref{Fig.PrePattern}, and hence this changes as each subsequent mode becomes permissible (i.e.~moving left to right each new mode has a larger value of $\lambda$). Hence we conjecture that if all modes are approximately excited by the same amount, then the envelope of unstable modes should scale with the fastest growing mode locally, which is qualitatively observed in Fig.~\ref{fig2}. Additionally, in the homogeneous setting, close to a supercritical bifurcation any patterned state should have an amplitude which scales with $\lambda$ raised to a power, and hence this provides an intuition for the final small-amplitude patterns observed in Figs.~\ref{fig1}-\ref{fig1b}, as again the envelope of the oscillations should scale with $\lambda$. However, we do not formally deduce a relationship between the envelope of the final patterned state with $\lambda$, and instead leave this as future work.

 \end{document}